\documentclass[manuscript,screen,acmsmall]{acmart}

\usepackage{amssymb}
\usepackage{mathtools}
\usepackage{stmaryrd}    
\usepackage{booktabs}
\usepackage{listings}
\usepackage{graphicx}

\acmJournal{TOSEM}
\acmYear{2026}
\acmVolume{0}
\acmNumber{0}
\acmArticle{0}
\acmMonth{7}
\setcopyright{none}          

\newcommand{\Av}{\mathcal{A}}                 
\newcommand{\Frag}{\mathcal{F}}               
\newcommand{\Full}{\mathcal{U}}               
\newcommand{\Time}{\mathsf{Time}}
\newcommand{\Val}{\mathsf{Val}}
\newcommand{\TVal}{\mathsf{TVal}}
\newcommand{\Pos}{\mathsf{Pos}}
\newcommand{\av}{\mathsf{av}}
\newcommand{\val}{\mathsf{val}}
\newcommand{\stampf}{\mathsf{stamp}}
\newcommand{\windowf}{\mathsf{window}}
\newcommand{\scanf}{\mathsf{scan}}
\newcommand{\joinf}{\mathsf{join}}
\newcommand{\resamplef}{\mathsf{resample}}
\newcommand{\asoff}{\mathsf{asof}}
\newcommand{\retrievef}{\mathsf{retrieve}}
\newcommand{\decidef}{\mathsf{decide}}
\newcommand{\letf}{\mathsf{let}}
\newcommand{\inkw}{\mathsf{in}}
\newcommand{\aggf}{\mathsf{agg}}
\newcommand{\predf}{\mathsf{pred}}
\newcommand{\Mono}{\mathsf{Mono}}
\DeclarePairedDelimiter{\dbrack}{\llbracket}{\rrbracket}
\newcommand{\dom}[1]{\dbrack{#1}}                  
\newcommand{\eff}{\varphi}                         
\newcommand{\rulename}[1]{\textsc{#1}}

\theoremstyle{acmplain}
\newtheorem{theorem}{Theorem}
\newtheorem{lemma}{Lemma}
\newtheorem{proposition}{Proposition}

\theoremstyle{acmdefinition}
\newtheorem{definition}{Definition}
\newtheorem{assumption}{Assumption}

\begin{document}

\title{Look-Ahead-Freedom as Temporal Non-Interference:
A Verifiable Correctness Property for Backtesting and Agentic
Trading Pipelines}

\author{Xavier Fonseca}
\affiliation{%
  \institution{Academy for AI, Games and Media, Breda University of
  Applied Sciences}
  \city{Breda}
  \country{The Netherlands}}
\email{xavier.fonseca.phd@gmail.com}

\begin{abstract}
Look-ahead bias---using information from after a decision epoch to make
the decision at that epoch---is the dominant way a backtest or a
machine-learning evaluation flatters a system that will disappoint in
deployment. The field manages it with construct-specific recipes and
empirical detectors, which are sound only channel by channel and certify
nothing by their silence. We show that look-ahead-freedom is a formal
property in disguise: fixing an epoch, the demand that the future not
influence the present is \emph{temporal non-interference} over a
time-indexed information lattice. From this identification we develop a
pipeline calculus separating a datum's availability from its reference
time, and settle the problem's boundary. Where availability may depend on
data values, look-ahead-freedom is undecidable ($\Pi^0_1$-hard): leakage
is recursively enumerable but freedom is not. On the value-independent
fragment---covering windowing, resampling, joins, point-in-time and
vintage reads, and agentic retrieval---we give a type-and-effect system
that is sound and decidable in linear time. An artifact confirms the
theory: the check scales linearly, an independent oracle witnesses no
leak in any accepted pipeline, and the checker catches every planted leak
that differential and tiling detectors miss.\footnote{\textbf{Reproducibility.} A public
artifact accompanies this paper, containing the checker, the baseline
detectors, the dynamic oracle, the adversarial corpus, and synthetic data
on which every qualitative result can be re-run. The headline figures
reported here are computed from proprietary market data that we are not
permitted to redistribute; the artifact reproduces the method and the
qualitative claims, not those exact figures. See the Data and Code
Availability statement (before the references).}
\end{abstract}



\maketitle

\section{Introduction}
\label{sec:intro}

A backtest is a claim about a decision that was never made: it reports
what a strategy \emph{would} have decided at each past instant, had it
been running then. The claim is only meaningful if each simulated
decision used exactly the information available at its own instant and
nothing from later. When this fails---when a computation performed for
time $t$ consumes a datum that only became known after $t$---the
backtest measures foresight the live system could never have had. In
quantitative finance this failure is called \emph{look-ahead bias}; in
machine learning it is one species of \emph{data leakage}. Under either
name it is the dominant way an evaluation flatters a system that will
disappoint in deployment, and it is notoriously easy to introduce and
hard to see: a normalisation fitted on the whole sample, a feature
aligned to the wrong timestamp, a restated fundamental read at its final
rather than its as-first-reported value, a universe assembled from the
firms that survived. Recent audits of machine-learning-based science and
of large-language-model trading agents find the problem to be pervasive
and, in the agentic case, subtler still: an agent can leak the future
through its pretraining corpus even when its code is fed only
contemporaneous inputs, so that inspecting the pipeline is not even
sufficient to rule leakage out~\cite{sarkar2025lookahead,ye2026alpha,li2025profit,exec2026assumptions}.

The field's response has been, almost exclusively, \emph{detection and
discipline}. Practitioners maintain catalogues of the ways a pipeline
can leak and matching recipes to avoid each---fit transforms on training
folds only, lag running statistics, split chronologically, embargo
around test windows---and a growing line of benchmarks measures, with
increasing sophistication, how much a given evaluation convention
inflates reported performance~\cite{kaufman2012leakage,benhenda2026lookahead,zhang2026alpha,zhang2026allleaks}.
This body of work is valuable and has sharpened the community's
awareness. But it shares a structural limitation. A recipe is sound only
for the specific channel it addresses and says nothing about the others;
a detector reports the leaks it happens to trigger and certifies nothing
by its silence. There is, at present, no \emph{property}---a statement a
pipeline either satisfies or violates, checkable once and for all---whose
satisfaction certifies that a pipeline is free of look-ahead of every
kind. The field manages the problem by vigilance where it could, in
principle, settle it by construction.

Such a property does exist in a neighbouring discipline, though it has
never been pointed at this problem. \emph{Non-interference}, the central
correctness notion of language-based information-flow
security~\cite{goguen1982security,denning1976lattice,volpano1996sound,sabelfeld2003language},
states that a program's low-security outputs are independent of its
high-security inputs: two runs agreeing on the low inputs produce
identical low outputs regardless of the high ones. It comes with a
mature enforcement technology (security type systems, effect systems)
and, in its timed branch, with a precisely mapped decidability
frontier---non-interference is undecidable for general timed automata
but decidable on identified
subclasses~\cite{gardey2007timed,andre2022guaranteeing}. The observation
that drives this paper is that look-ahead-freedom \emph{is} an instance
of this notion, once the security lattice is replaced by a time-indexed
one. Fixing a decision epoch $t$, the requirement ``information from
times $t' > t$ must not influence the decision at $t$'' is exactly a
non-interference statement in which the protected (``high'') partition is
the future relative to $t$, the observable (``low'') output is the
decision at $t$, and the ordering is temporal availability rather than
clearance. Look-ahead bias is temporal non-interference.

Taking this identification seriously turns an informally managed hazard
into a formal object with the full apparatus of information-flow theory
behind it, and the theory immediately tells us both the bad news and the
shape of the good. The bad news is a hard limit: on a pipeline language
expressive enough to let a datum's \emph{availability} be computed from
other data's \emph{values}, look-ahead-freedom is undecidable---indeed
$\Pi^0_1$-hard---so no analyser can certify freedom for every pipeline
one can write (Theorem~\ref{thm:undecidable}). Leakage is
recursively enumerable (a witnessing future-dependent run can be found
if one exists) but freedom is not, which is precisely why detectors can
exhibit leaks yet never certify their absence. The good news is that the
undecidability is caused by a feature real pipelines do not use.
Backtests, feature constructions, and agentic retrieval steps determine
\emph{when} each datum is available from its schema and schedule, not
from the numeric values flowing through the pipeline. Isolating this as
a syntactic condition yields a fragment---value-independent
availability---on which look-ahead-freedom is decidable, and on which a
static discipline can certify it soundly.

\paragraph{Contributions.} This paper develops that programme in full.

\begin{itemize}
\item We give a time-indexed pipeline calculus (Section~\ref{sec:calculus})
  that separates a datum's \emph{availability} (when it can be read) from
  its \emph{reference time} (which instant it describes), with a
  two-level grammar whose availability sublanguage is value-independent
  and whose membership in the decidable fragment is checkable in linear
  time. We define \emph{look-ahead-freedom} as temporal non-interference
  over the resulting time-indexed information lattice
  (Definition~\ref{def:laf}), via an exact two-run agreement condition
  on an epoch-parametrised operational semantics
  (Section~\ref{sec:semantics}).

\item We give a type-and-effect system for the fragment and prove it
  \emph{sound}: every pipeline it accepts is free of temporal leakage
  (Theorem~\ref{thm:soundness}), with acceptance decidable in time
  linear in the pipeline's size under bounded availability-term
  complexity (Theorem~\ref{thm:decidable-check}). The proof is a two-run
  logical relation whose crux is the treatment of re-stamping, the
  channel by which a future value can be relabelled to an admissible
  time; we characterise the analysis's precision exactly, including the
  opaque-value-operation cases where a sound analysis must
  conservatively reject a semantically clean pipeline
  (Proposition~\ref{prop:rel-complete}). We further harden the treatment
  of the series operators---windowing, scans, joins, and
  resampling---and isolate non-causal resample/join alignment as the
  genuine leakage channel there, discharged by an explicit $O(1)$ checker
  obligation (Section~\ref{sec:series-hardening}).

\item We establish the boundary that makes the fragment necessary rather
  than merely convenient: on the full language look-ahead-freedom is
  undecidable and $\Pi^0_1$-hard, by reduction from the halting problem
  through data-dependent availability (Theorem~\ref{thm:undecidable}).
  Detection is r.e.; certification of freedom is not.

\item We validate the checker empirically on a public, clean-room
  artifact (Section~\ref{sec:evaluation}). The check scales linearly in
  pipeline size (measured log--log slope $1.023$, 95\% bootstrap CI
  $[1.021, 1.026]$, across four orders of magnitude); on five archetypes
  drawn from proprietary market data (not redistributable) an independent dynamic oracle
  witnesses no leak in any accepted pipeline, corroborating soundness;
  and on an adversarial corpus the sound checker misses none of $33$
  planted leaks while a differential two-run detector of the kind used in
  practice misses $18$ ($54.5\%$) and a window-tiling detector misses all
  $33$, with the checker's cost a precisely characterised and confined
  set of false positives on opaque value operations.
\end{itemize}

The through-line is a single conversion. The empirical detectors answer
``did this pipeline, on the data and epochs I tried, exhibit a leak?'';
the property developed here answers ``can this pipeline, on any data and
at any epoch, exhibit a leak?'', and answers it soundly and---on the
pipelines practitioners actually write---decidably and in linear time.
It moves look-ahead-freedom from something a team hopes it has achieved
into something a checker can certify.

\section{Background and Related Work}
\label{sec:background}

Our contribution sits at the intersection of three literatures that
have, to date, developed independently: (i) the applied treatment of
\emph{look-ahead bias} and \emph{data leakage} in quantitative finance
and machine learning; (ii) the theory of \emph{information-flow
security} and \emph{non-interference} in programming languages and
timed systems; and (iii) \emph{data provenance} and \emph{stream /
dataflow correctness} in data-engineering systems. We review each in
turn, then state precisely the gap that none of them closes and that
this paper addresses: a formal, verifiable \emph{property}---as opposed
to a checklist, an audit, or an empirical detector---certifying that a
computational pipeline does not allow information from the future to
influence a decision made in the past. The stance mirrors a recurring
lesson in decision-focused evaluation, where the quantity that governs a
decision's correctness is not the one conventionally
measured~\cite{fonseca2026decision}.

\subsection{Look-Ahead Bias and Data Leakage as an Applied Problem}
\label{sec:bg-leakage}

Look-ahead bias---using information not available at decision time---is
recognised across quantitative finance and machine learning as a
pervasive threat to the validity of reported results. In the
backtesting and portfolio literature it is treated as one of a family
of evaluation failures alongside survivorship
bias~\cite{brown1992survivorship}, backtest
overfitting~\cite{bailey2014pseudo,bailey2014deflated,bailey2017probability}, transaction-cost
neglect, and regime-shift blindness. The
rapid migration of large language models into trading systems has
sharpened the problem: a taxonomy-oriented evaluation study of
financial multi-agent systems argues that coordination and evaluation
protocol, rather than model scale, drive measured performance and that
claims across systems are hard to compare~\cite{eval2026reliable}, and
an audit-oriented evidence map of LLM trading agents finds acute
\emph{protocol incomparability}, reporting that only a small minority of
surveyed empirical studies document time-consistent train/test split
protocols, explicit transaction-cost models, or survivorship handling,
and that none reach the highest tier of
reproducibility~\cite{xia2026agentic}. Work on execution assumptions and
reproducibility in LLM-based trading makes the mechanism concrete,
emphasising that mundane choices---executing at the signal-generating
close versus the next tradable price, forming a universe from future
constituents, omitting costs---can move a strategy from apparently
profitable to unusable, and that such retrospective leakage can be
severe enough to overturn a strategy's apparent
edge~\cite{exec2026assumptions}. For agentic pipelines the channel is
subtler still. An agent can leak future information through its
pretraining data or a retrieval corpus even when its code is fed only
contemporaneous inputs---a form of look-ahead that has no analogue in
classical backtesting and that inspection of the pipeline code cannot
rule out~\cite{sarkar2025lookahead}. Reported returns then need not
reflect a causal pipeline at all~\cite{ye2026alpha}, and a growing line
of benchmarks quantifies exactly this decay once a model's knowledge
window ends~\cite{li2025profit}, some at claim-level granularity, using
Shapley attribution to isolate which parts of a rationale draw on leaked,
post-cutoff information~\cite{zhang2026allleaks}. Dedicated benchmarks
have begun to \emph{measure} look-ahead bias directly in point-in-time
language models, quantifying alpha decay across temporally distinct
regimes~\cite{benhenda2026lookahead}, and paired one-switch benchmarks
isolate how individual evaluation conventions inflate backtested
performance~\cite{zhang2026alpha}. All of this work supplies empirical
measurement; none supplies the formal property this paper provides.

The same problem is described, in almost identical terms, far outside
finance. In network-intrusion detection, temporal leakage is addressed
through a catalogue of mitigation strategies---per-split normalisation
computed only on past data, lagged running statistics, causal masking of
attention, and chronological rather than random
splits~\cite{nids2026survey}. In Android-malware detection, temporal
leakage is distinguished taxonomically from train--test leakage and
identified as always an experiment-design flaw, with sliding-window
datasets offered as the remedy~\cite{android2024leakage}. Practitioner
and vendor guidance converges on the same recipes: encapsulating
preprocessing so that fit/transform never sees the test partition, and
using time-series cross-validation for temporally dependent
data~\cite{kaufman2012leakage}.

Two properties characterise this entire body of work. First, it is
overwhelmingly \emph{informal}: leakage is defined by example and
addressed by discipline, tooling conventions, and vigilance rather than
by a stated property that a pipeline either satisfies or violates.
Second, its remedies are \emph{sound only by construction and only for
the specific construct addressed}---per-split normalisation prevents one
leakage channel but says nothing about others, and no combination of
such recipes yields a guarantee that a given pipeline is leakage-free.
The field has, in effect, a growing list of ways to leak and a matching
list of ways to avoid each, but no notion of a certificate of absence.
The closest step toward a stated property casts point-in-time
correctness in filtration terms---requiring that every decision's
information set lie within the natural filtration
$\mathcal{F}_t$~\cite{fonseca2026pit}---but does so for a fixed family of
trading strategies rather than as a decidable discipline over a general
pipeline language, which is the gap this paper closes.

\subsection{Information-Flow Security and Non-Interference}
\label{sec:bg-noninterference}

The property the applied literature lacks has, in a different domain, a
mature theory. \emph{Non-interference} formalises information-flow
security: partitioning a program's variables into high- and
low-security classes, a program is non-interferent if its low-security
outputs are independent of its high-security inputs---equivalently, two
executions agreeing on low inputs produce identical low outputs
regardless of high inputs~\cite{goguen1982security,volpano1996sound}. The security levels are organised as a lattice, following Denning~\cite{denning1976lattice}, and secure flow can be certified statically against that lattice~\cite{denning1977certification}. The
subsequent three decades of language-based enforcement are surveyed by
Sabelfeld and Myers~\cite{sabelfeld2003language}. The property is
enforced by information-flow type systems, by taint-tracking, and, at
the hardware level, by gate-level information-flow tracking, which
establishes sufficient conditions under which a system provably
satisfies non-interference~\cite{tiwari2009complete}.

Crucially for our purposes, the theory has an explicitly \emph{temporal}
branch. \emph{Timed} notions of non-interference have been defined on
timed automata, where it is known that adding timing constraints can
turn a secure system insecure, and---decisively for the shape of our
contribution---that the general problem of deciding timed strong
non-deterministic non-interference is \emph{undecidable}, while specific
cosimulation-, bisimulation-, and state-based variants are
\emph{decidable}~\cite{gardey2007timed}. This ``undecidable in general,
decidable on a subclass'' pattern recurs in the parametric-timed-opacity
setting of Andr\'e et al.~\cite{andre2022guaranteeing}, whose undecidability
result for general parametric timed automata (via a
reachability/halting reduction) and decidability result for the L/U-PTA
subclass we take as methodological precedent for locating our own
boundary. Related lines develop complexity-oriented ``temporal
non-interference'' by transposing information-flow security into the
timed-automata setting~\cite{alur1994theory}, where the secret is
\emph{when} an action occurs and decidability sits on a
knife-edge~\cite{gardey2007timed,andre2022guaranteeing}.

This literature supplies exactly the missing apparatus: a definition
schema (independence of a protected partition), enforcement mechanisms
(type systems, model checking), and---most valuable---a known
decidability frontier telling us in advance that the general temporal
problem is intractable and that the research contribution lies in
isolating a decidable fragment that covers the pipelines of interest.
What it has never been applied to is look-ahead bias: the ``high''
partition in every prior treatment is a secret to be protected from an
adversary, never \emph{data from the future} to be withheld from a
decision made in the past.

\subsection{Provenance and Dataflow / Stream Correctness}
\label{sec:bg-provenance}

A third literature reasons about how data moves through computational
pipelines. Scientific-workflow \emph{provenance} records the derivation
of results, distinguishing prospective provenance (the workflow
specification) from retrospective provenance (the executed run), and
represents causality as a dependency graph over artefacts and processes;
the community has standardised these notions in OPM and W3C PROV, and
separates why- from where-provenance~\cite{cheney2009provenance}.
Provenance thus records \emph{which} inputs produced an output, but it
carries no notion of a time-indexed admissibility condition: it can tell
us an output depended on a given input, but not that the input was drawn
from a timestamp later than the output's decision time, nor does it
define a violation. Even where provenance is extended toward \emph{temporal attribution}
---quantifying which inputs influenced an output over time---it
remains a descriptive record, not a correctness guarantee: it reports
what did flow, not what \emph{may} flow across the availability
boundary.

The stream-processing and dataflow-language communities come closest to
our concern. The dataflow model formalises correct processing of
unbounded, out-of-order data using watermarks and timers to reason about
event-time completeness~\cite{akidau2015dataflow}. Most pointedly, the
Causify DataFlow framework for high-performance ML stream computing
states our exact hazard---that computations at time $t$ may access
observations from a later time $t' > t$, which it calls
future-peeking---and builds a substantial formal treatment around
it~\cite{causify2025dataflow}. It defines \emph{causal computation} (a
node's output at simulation time depends only on data timestamped
earlier) and \emph{point-in-time idempotency} (once the input window is
at least a minimal context length $L$, the output at $t$ is independent
of how far back the window starts), and it \emph{proves} from the latter
a chain of \emph{tilability} and tile-correctness results guaranteeing
that batch execution equals streaming execution.

The decisive distinction is what receives a proof and what does not.
Point-in-time idempotency and tilability are established as theorems and
enforced by construction. \emph{Causality itself}---the future-peeking
property that is our concern---is, in the framework's own terms,
\emph{detected}: violations are surfaced through a testing and replay
framework and validated by tiling tests that partition the data in
multiple ways, a differential-testing oracle that compares batch against
streaming (or across tilings) and flags a leak when outputs disagree.
There is no soundness theorem certifying a computation causal, and no
characterisation of when such certification is possible.

This is the closest existing notion to look-ahead-freedom, and we
position against it precisely. First, DataFlow's causality check is an
\emph{empirical detector}: it exposes a violation when a chosen set of
tile lengths triggers an inconsistency, but certifies nothing when none
is observed---absence of detection is not proof of absence. We give
instead a \emph{sound} static property: an accepted pipeline is provably
look-ahead-free (Theorem~\ref{thm:soundness}). Second, we explain
\emph{why} their causality story is a test rather than a proof:
certifying look-ahead-freedom is not recursively enumerable on the full
language (Theorem~\ref{thm:undecidable}), so differential testing is, in
a precise sense, the best a general-purpose method can do---and a sound
certificate is available only on a decidable fragment, which we identify.
Third, our treatment of causal windows is the static-effect analogue of
their point-in-time idempotency: where they establish a minimal context
length $L$ semantically per node and validate it by tiling, our
type-and-effect discipline bounds availability statically and checks it
without execution. The two are kin; the contribution is the shift from a
tested, per-node semantic condition to a sound, decidable, whole-pipeline
static guarantee, together with the boundary result that says no better
is possible in general.

\subsection{The Gap and Our Position}
\label{sec:bg-gap}

Across all three literatures the same shape recurs and remains unfilled.
The applied leakage literature has the \emph{problem} but treats it
informally, with construct-specific recipes and no certificate of
absence. The information-flow literature has the \emph{property form and
enforcement machinery}---including a temporal branch with a known
decidability frontier---but has only ever applied it to adversary-facing
secrecy, never to temporal data availability. The provenance and
dataflow literature has the \emph{pipeline substrate} and, in its most
recent stream-computing instance, an empirical \emph{detector} for
future-to-past information use, but no formal property and no guarantee.

We close this gap by making one identification precise and exploiting
it. \textbf{Look-ahead bias is temporal non-interference.} Fixing a
decision time $t$, ``information from times $t' > t$ must not influence
the decision at $t$'' is exactly a non-interference statement in which
the protected (``high'') partition is the future relative to $t$ and the
observable (``low'') output is the decision at $t$, taken over a
time-indexed information lattice rather than a security lattice.
Leakage-freeness of a backtest, feature-construction, or agentic trading
pipeline is then the statement that the pipeline satisfies temporal
non-interference with respect to this time-indexed partition.

From this identification the paper delivers what the three literatures
individually cannot: a formal definition of \emph{look-ahead-freedom} as
temporal non-interference (Definition~\ref{def:laf}); the decidability
boundary that separates the pipelines one can certify from those one
cannot (Theorems~\ref{thm:soundness} and~\ref{thm:undecidable}); a sound
static checker for the former; and an empirical validation on real
backtesting and agentic-trading pipelines
(Section~\ref{sec:evaluation}). The detailed contributions are set out in
Section~\ref{sec:intro}. The upshot is a single conversion---a problem
the field currently manages by vigilance becomes one it can certify by
construction.

\section{A Time-Indexed Pipeline Calculus}
\label{sec:calculus}

We now make the identification of Section~\ref{sec:background} precise.
This section fixes the object of study---a small calculus of
time-indexed data pipelines---and isolates the syntactic fragment on
which look-ahead-freedom will turn out to be decidable. The calculus is
deliberately minimal: it contains exactly the constructs that recur in
backtesting, feature construction, and agentic-retrieval pipelines
(windowing, causal folds, resampling, point-in-time reads, and
context retrieval), and nothing else. Section~\ref{sec:typesystem} then
gives a static discipline over this calculus and proves it sound and
decidable; Section~\ref{sec:decidability} shows that the same property
is undecidable once the fragment's defining restriction is lifted,
locating the exact boundary.

The design has one central idea, which we state before the syntax so the
grammar reads as its consequence. Every datum in a pipeline carries two
independent pieces of temporal information: \emph{what} value it holds,
and \emph{when} that value becomes knowable---its \emph{availability}.
A pipeline is free of look-ahead exactly when no value that becomes
knowable after a decision epoch $t$ can influence the decision made at
$t$. The calculus therefore separates the two pieces at the level of
syntax: a datum's value may depend on other data freely, but in the
decidable fragment a datum's \emph{availability} may not depend on any
data value. This separation is what makes availability a statically
tractable quantity, and it is justified by how real pipelines are
built---availability is exogenous, fixed by the data source
(a publication date, a filing date, a bar close), ingested immediately
and unconditionally; a well-formed pipeline never re-times a datum's
availability on the basis of another datum's value. We make the
separation precise as a two-level grammar and then show that lifting it
is the sole source of undecidability.

\begin{figure}[t]
  \centering
  \includegraphics[width=\linewidth]{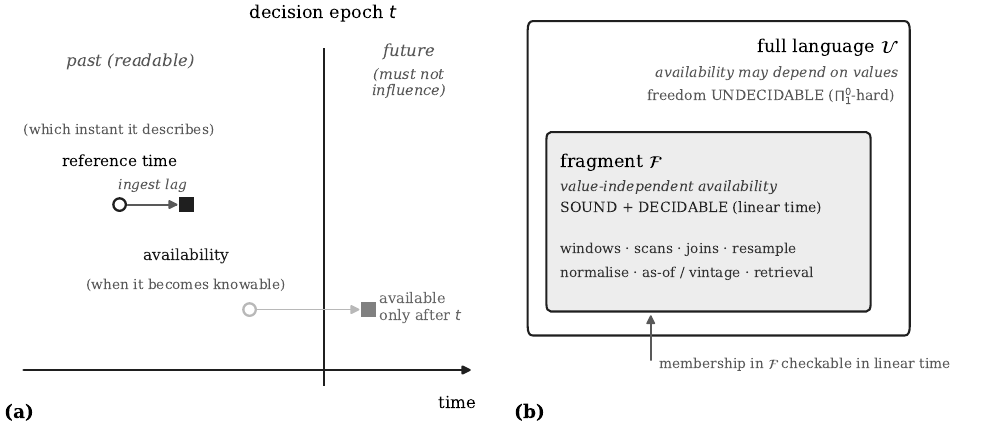}
  \caption{The two ideas the calculus is built on. (a) Every datum
  carries two independent temporal coordinates: its \emph{reference time}
  (which instant it describes) and its \emph{availability} (when it
  becomes knowable); the two differ by an ingest lag, and look-ahead
  occurs precisely when a datum whose availability follows a decision
  epoch $t$ influences the decision at $t$. (b) The language is stratified
  by whether availability may depend on data values. On the full language
  $\Full$, availability can be value-dependent and look-ahead-freedom is
  undecidable ($\Pi^0_1$-hard, Theorem~\ref{thm:undecidable}); on the
  value-independent fragment $\Frag$---which covers the constructs real
  pipelines use---the property is sound and decidable in linear time
  (Theorems~\ref{thm:soundness} and~\ref{thm:decidable-check}), and
  membership in $\Frag$ is itself checkable in linear time.}
  \label{fig:concept}
\end{figure}

\subsection{Sorts and Time-Indexed Values}
\label{sec:calc-sorts}

We work over a time domain $(\mathbb{T}, \le)$, a totally ordered set of
epochs. The calculus has three sorts,
\[
  \tau ::= \Time \;\mid\; \Val \;\mid\; \TVal,
\]
where $\Time$ ranges over $\mathbb{T}$, $\Val$ over ordinary (untimed)
data values, and $\TVal$ over \emph{time-indexed values} $\langle
v,\alpha\rangle$---a $\Val$ $v$ paired with its availability stamp
$\alpha\in\mathbb{T}$. The stamp $\alpha$ records when $v$ becomes
knowable, not when it is about; these can differ, and the distinction
between a datum's availability and its reference time is exactly what
look-ahead reasoning must track. A decision made at epoch $t$ may
legitimately consume a time-indexed value $\langle v,\alpha\rangle$ only
when $\alpha \le t$.

\subsection{The Availability Sublanguage}
\label{sec:calc-avail}

Availability terms denote elements of $\Time$ and are built only from
other availability information---never from values:
\[
\begin{array}{rcll}
  A &::=& t                         & \text{(the current decision epoch)}\\
    &\mid& c                        & \text{(a constant stamp, } c\in\mathbb{T})\\
    &\mid& \av(x)                   & \text{(availability of a bound time-indexed variable } x)\\
    &\mid& A + \delta               & \text{(constant lag/lead, } \delta\in\mathbb{Z})\\
    &\mid& \max(A_1,\dots,A_n)      & \text{(latest-of)}\\
    &\mid& \min(A_1,\dots,A_n)      & \text{(earliest-of, fixed arity } n).
\end{array}
\]
We write $\Av$ for the set of availability terms generated by this
grammar. The defining discipline is that the only way a datum enters an
availability term is through $\av(x)$, which extracts the \emph{stamp}
of a time-indexed variable, never its value $\val(x)$. Consequently
every term of $\Av$ is a function of stamps and constants alone. The
arity $n$ of $\max$ and $\min$ is fixed and syntactically apparent:
there is no folding over a data-dependent collection, which is what
prevents an availability term from smuggling in a data-dependent number
of inputs. This restriction is what the following lemma records, and it
is the linchpin of the entire development.

\begin{lemma}[Value-independence of availability]
\label{lem:value-indep}
For every availability term $A\in\Av$ with free time-indexed variables
$x_1,\dots,x_k$, the denotation of $A$ depends on $x_1,\dots,x_k$ only
through $\av(x_1),\allowbreak\ \dots,\allowbreak\ \av(x_k)$; it is invariant under any change to
the values $\val(x_i)$.
\end{lemma}
\begin{proof}
Structural induction on $A$. No production of the grammar mentions
$\val(\cdot)$, and the only leaf that touches a variable is $\av(x)$,
which reads its stamp. Every other leaf is $t$ or a constant, and the
internal nodes $+\delta$, $\max$, $\min$ are functions of their
subterms' denotations. Hence the denotation factors through the stamps
alone.
\end{proof}

Lemma~\ref{lem:value-indep} holds \emph{by grammar}, and membership in
$\Av$ is therefore a syntactic, decidable predicate. This is what will
let the enforcement mechanism of Section~\ref{sec:typesystem} treat
every availability as a statically known quantity.

\subsection{Pipeline Terms}
\label{sec:calc-pipeline}

Pipeline terms denote time-indexed values or \emph{series} of them. For
the series-valued constructs we distinguish two orderings that the
scalar fragment does not need to separate: a \emph{position}---an
element's index within a series, ordered by $\preceq$---and its
availability time $\av\in\mathbb{T}$, ordered by $\le$. For a single
time-indexed value the two coincide operationally; the position domain
$(\Pos,\preceq)$ and the availability-monotonicity invariant that links
$\preceq$ to $\le$ are introduced with the series operators in
Section~\ref{sec:series-hardening}, where they are needed. The grammar
is stated once:
\[
\begin{array}{rcll}
  e &::=& x                                   & \text{(variable)}\\
    &\mid& \langle d, A\rangle                & \text{(base datum: value } d \text{ available at stamp } A)\\
    &\mid& \val(e)                            & \text{(project the value)}\\
    &\mid& f(e_1,\dots,e_n)                   & \text{(pure value op)}\\
    &\mid& \stampf(e, A)                      & \text{(attach/replace availability with } A)\\
    &\mid& \windowf_{k}(e)                    & \text{(last-}k\text{ window at each epoch)}\\
    &\mid& \scanf(f, e_0, e)                  & \text{(causal fold)}\\
    &\mid& \joinf(e_1,\dots,e_n)              & \text{(align sources at an epoch)}\\
    &\mid& \resamplef_{\rho}(e)               & \text{(down/upsample by rule } \rho)\\
    &\mid& \asoff(e_{\text{store}}, A)        & \text{(point-in-time / vintage read as of } A)\\
    &\mid& \retrievef(e_{\text{ctx}}, A)      & \text{(agentic context fetch as of } A)\\
    &\mid& \decidef_{t}(e)                    & \text{(emit a decision at epoch } t)\\
    &\mid& \letf\ x = e_1\ \inkw\ e_2.
\end{array}
\]

Two features of the grammar carry the availability discipline
explicitly. First, every construct that produces a fresh availability
stamp---$\stampf$, $\joinf$, $\resamplef$, $\asoff$, and
$\retrievef$---takes its new stamp from an availability term $A$. In the
fragment (Section~\ref{sec:calc-fragment}), $A\in\Av$, so the produced
stamp is value-independent. Their natural stamps are the expected ones:
$\joinf$ produces the $\max$ of its input stamps, and
$\asoff(\cdot,A)$ and $\retrievef(\cdot,A)$ produce stamp $A$, typically
$t$ or a source timestamp. Second, $\asoff$ and $\retrievef$ are the
point-in-time (vintage) and agentic-context constructs: their returned
\emph{value} may depend on the store's contents and on $A$, but their
\emph{stamp} is $A$, which in the fragment is value-independent. This is
precisely why vintages and restatements---the subtlest case in
practice---remain on the decidable side of the boundary. Reading a store \emph{as of} a reference time is a controlled release of data that is otherwise inaccessible at the decision epoch, i.e.\ a declassification with an explicit \emph{when} dimension~\cite{sabelfeld2009declassification}.

The one assumption we make about value operations isolates them from the
temporal analysis entirely.

\begin{assumption}[Value operations are total and pure]
\label{ass:pure}
Every value operation $f$, both in $f(e_1,\dots,e_n)$ and as the fold
function of $\scanf$, is a total, pure function of its $\Val$ arguments:
it always returns a value, and it reads no data other than those
arguments---in particular it performs no side-channel access to base
data, the store, or availability information. Consequently $f$ can
influence a result only through the values explicitly passed to it, and
its output is determined by those values alone.
\end{assumption}

Assumption~\ref{ass:pure} is what licenses treating a value operation as
\emph{opaque}: the analysis of Section~\ref{sec:typesystem} neither
needs nor is permitted to look inside $f$, because $f$ cannot introduce a
dependency that its arguments do not already carry. It is also the exact
hypothesis under which the relative-completeness result
(Proposition~\ref{prop:rel-complete}) is stated.

\subsection{The Fragment and the Full Language}
\label{sec:calc-fragment}

\begin{definition}[Value-independent fragment $\Frag$]
\label{def:fragment}
A pipeline term $e$ is in the \emph{fragment} $\Frag$ if and only if
every availability term occurring in it---in a base datum, $\stampf$,
$\joinf$, $\resamplef$, $\asoff$, or $\retrievef$---is drawn from the
sublanguage $\Av$ of Section~\ref{sec:calc-avail}; that is, no
availability is produced by a general pipeline term.
\end{definition}

\begin{definition}[Full language $\Full$]
\label{def:full}
The \emph{full language} $\Full$ replaces the availability argument of
the stamping constructs with an arbitrary pipeline term of sort $\Time$,
\[
  \stampf(e, e_A),\quad \asoff(e_{\text{store}}, e_A),\ \dots,
  \qquad e_A : \Time,
\]
so that an availability may be computed from data values---for example,
$e_A = \mathsf{if}\ \val(e_1) > c\ \mathsf{then}\ t-2\ \mathsf{else}\ t$,
modelling endogenous, value-conditioned availability. This case lies
outside the pipelines that motivate the work, but we include it so the
theory \emph{characterises} the boundary rather than assuming it away.
\end{definition}

\begin{proposition}[Fragment membership is decidable]
\label{prop:membership}
$e \in \Frag$ is decidable in time linear in the size of $e$.
\end{proposition}
\begin{proof}
Traverse the term once; at each construct carrying an availability
argument, check that the argument parses in the grammar $\Av$, which
forbids $\val$ and data-dependent arity. This is a single syntactic
pass.
\end{proof}

Proposition~\ref{prop:membership} makes the eventual checker honest
about its own scope: it can always tell, cheaply, whether a submitted
pipeline lies in the region where soundness is guaranteed ($\Frag$) or
in the general region where, by Section~\ref{sec:decidability}, no sound
and complete checker can exist, and where it must therefore
conservatively reject or warn.

We defer the operational semantics---the execution that advances $t$
through $\mathbb{T}$---and the formal definition of look-ahead-freedom as
temporal non-interference to Section~\ref{sec:semantics}. Informally,
and sufficiently for reading the type system: an execution is
\emph{look-ahead-free} if, for every decision emitted at epoch $t$,
perturbing any base datum whose availability exceeds $t$ leaves that
decision unchanged. This is the two-run, non-interference form of the
property, and it is the statement Theorem~\ref{thm:soundness} certifies.

\section{A Sound, Decidable Discipline for the Fragment}
\label{sec:typesystem}

We give a type-and-effect system $\vdash_{\Frag}$ for terms of the
fragment $\Frag$ and prove two results about it.

\begin{theorem}[Soundness]
\label{thm:soundness}
If $\;\vdash_{\Frag} e$, then $e$ is look-ahead-free
(Definition~\ref{def:laf}).
\end{theorem}

\begin{theorem}[Decidability]
\label{thm:decidable-check}
$\;\vdash_{\Frag} e$ is decidable in time linear in $|e|$.
\end{theorem}

Together with the undecidability of the property on the full language
$\Full$ (Theorem~\ref{thm:undecidable}, Section~\ref{sec:decidability}),
these results locate the exact boundary: look-ahead-freedom is
undecidable in general, but admits a sound, linear-time static check on
the value-independent fragment that covers real point-in-time pipelines.

\subsection{Effects as Static Availability Bounds}
\label{sec:effects}

The \emph{effect} of a term is a conservative upper bound, expressed as
an availability term, on the availability of every base datum whose
\emph{value} can flow into the term's result. This is a type-and-effect
discipline in the sense of Lucassen and
Gifford~\cite{lucassen1988polymorphic,gifford1986integrating} (see~\cite{nielson1999principles} for the general theory): the type
records what a term computes, and the effect records a statically
computable, conservative over-approximation of a runtime quantity---here,
the temporal availability of the data the result depends on. Effects live in the
availability sublanguage $\Av$, so they are themselves value-independent
and statically evaluable.

\begin{definition}[Effect ordering]
\label{def:effect-order}
For availability terms $\eff,\psi \in \Av$, write $\eff \sqsubseteq
\psi$ if and only if $\dom{\eff}_{t,I} \le \dom{\psi}_{t,I}$ for all
epochs $t$ and all stamp valuations $I$. Because terms of $\Av$ are
monotone functions of $\{t,\text{ stamps},\text{ constants}\}$ built
from $\{+\delta,\max,\min\}$, the relation $\sqsubseteq$ is a partial
order with $\max$ as join, and it is decidable
(Lemma~\ref{lem:effect-order}).
\end{definition}

\begin{lemma}[Decidability of effect ordering]
\label{lem:effect-order}
$\eff \sqsubseteq \psi$ is decidable for $\eff,\psi\in\Av$.
\end{lemma}
\begin{proof}
Normalise each of $\eff,\psi$ to the form
$\max_i\big(\min_j(t + c_{ij}\ \text{or}\ \av(x_k)+c_{ij})\big)$ by
distributing $+\delta$ over $\max$ and $\min$ (all shifts are constant)
and flattening; arities are fixed, so normalisation terminates.
Comparing two such normal forms reduces to a finite set of
constant-and-stamp inequalities that must hold for all $t$ and all stamp
valuations. Each variable ($t$ and each $\av(x)$) occurs with unit
coefficient, so the comparison is decided by the constant offsets
together with the $\max$/$\min$ tree structure.
\end{proof}

\subsection{The Typing Judgement}
\label{sec:typing}

The judgement has the form $\Gamma \vdash e : \tau \,!\, \eff$, read:
under context $\Gamma$, the term $e$ has sort $\tau$ and effect
$\eff\in\Av$---every base-datum value influencing $e$'s result is
available no later than $\eff$. The context $\Gamma$ maps each variable
to its sort and its stamp term (an element of $\Av$). We give the rules
that carry temporal content; the omitted congruence rules are standard
and simply join effects with $\max$.

The base and projection rules record where a value comes from:
\[
\frac{}{\Gamma \vdash \langle d, A\rangle : \TVal\,!\, A}
\quad(\rulename{T-Base})
\qquad\qquad
\frac{\Gamma \vdash e : \TVal\,!\,\eff}
     {\Gamma \vdash \val(e) : \Val\,!\,\eff}
\quad(\rulename{T-Val})
\]
A base datum's value is available exactly at its stamp $A$; that is its
effect. Projecting a value inherits the effect, since reading the
content exposes whatever availability produced it. Pure value
operations combine effects by $\max$:
\[
\frac{\Gamma \vdash e_i : \Val\,!\,\eff_i \quad (\forall i)}
     {\Gamma \vdash f(e_1,\dots,e_n) : \Val\,!\,\textstyle\max_i \eff_i}
\quad(\rulename{T-Op})
\]
the result can depend on the latest-available of its operands.

The next rule is the crux of the system. Re-stamping a datum changes its
advertised stamp but must not change its effect:
\[
\frac{\Gamma \vdash e : \Val\,!\,\eff \qquad A \in \Av}
     {\Gamma \vdash \stampf(e, A) : \TVal\,!\,\eff}
\quad(\rulename{T-Stamp})
\]
The value produced by $\stampf(e,A)$ still depends on data available at
$\eff$, even though it now advertises the stamp $A$. This is exactly
where a stamp-only check fails and the effect system succeeds:
$\stampf(e,A)$ may present an admissible stamp $A\le t$ while its effect
$\eff \not\sqsubseteq t$ betrays a future dependency. A check that reads
only the advertised stamp would accept such a term; the effect discipline
does not.

The series and join constructs combine their premises' effects:
\[
\frac{\Gamma \vdash e_i : \TVal\,!\,\eff_i \quad(\forall i)}
     {\Gamma \vdash \joinf(e_1,\dots,e_n) : \TVal\,!\,\textstyle\max_i \eff_i}
\quad(\rulename{T-Join})
\qquad
\frac{\Gamma \vdash e : \TVal\,!\,\eff}
     {\Gamma \vdash \windowf_k(e) : \TVal\,!\,\eff}
\quad(\rulename{T-Window})
\]
A causal last-$k$ window over an availability-monotone series has effect
equal to the newest element's, namely $\eff$: the window never reaches
past the current position, as its causal semantics
(Section~\ref{sec:series-hardening}) guarantee. The fold rule is
\[
\frac{\Gamma \vdash e_0 : \TVal\,!\,\eff_0 \qquad
      \Gamma \vdash e : \TVal\,!\,\eff}
     {\Gamma \vdash \scanf(f,e_0,e) : \TVal\,!\, \max(\eff_0,\eff)}
\quad(\rulename{T-Scan})
\]
A causal fold's running result at position $p$ depends on the seed $e_0$
and on inputs up to $p$. By monotonicity
(Lemma~\ref{lem:mono-preserve}), the current element's effect $\eff$
dominates all earlier inputs, so the only term the current input does
not necessarily dominate is the seed effect $\eff_0$; the sound bound is
therefore $\max(\eff_0,\eff)$, and the purity of $f$
(Assumption~\ref{ass:pure}) guarantees the accumulator introduces no
availability beyond these. When the seed is admissible at the series
start, the bound simplifies to $\eff$ at every non-initial position; we
keep the explicit $\max(\eff_0,\eff)$ so that soundness requires no
assumption on the seed.

The vintage rule is the second decisive one:
\[
\frac{\Gamma \vdash e_{\text{store}} : \TVal\,!\,\eff_s \qquad A\in \Av}
     {\Gamma \vdash \asoff(e_{\text{store}}, A) : \TVal\,!\, A}
\quad(\rulename{T-Asof})
\]
A point-in-time read returns the store content as known at $A$; its
effect is $A$, \emph{not} $\eff_s$, because the as-of semantics
(Section~\ref{sec:semantics}) guarantee that only data with availability
$\le A$ is visible to the read. This is the formal reason vintages and
restatements stay decidable: the value depends on the store, but only on
its $\le A$ portion, and $A\in\Av$ is value-independent. The agentic
retrieval construct $\retrievef$ types identically.

Finally, the decision rule carries the single obligation that enforces
the property, and $\letf$ substitutes a binding's effect into its body:
\[
\frac{\Gamma \vdash e : \tau\,!\,\eff \qquad \eff \sqsubseteq t}
     {\Gamma \vdash \decidef_t(e) : \tau\,!\, t}
\quad(\rulename{T-Decide})
\qquad
\frac{\Gamma \vdash e_1 : \tau_1\,!\,\eff_1 \quad
      \Gamma, x{:}\tau_1 \vdash e_2 : \tau_2\,!\,\eff_2}
     {\Gamma \vdash \letf\ x = e_1\ \inkw\ e_2 : \tau_2 \,!\, \eff_2[\eff_1/\av(x)]}
\quad(\rulename{T-Let})
\]
A decision at epoch $t$ is well-typed only if the effect of its body is
bounded by $t$: every value that can flow into it is available by $t$.
This single premise $\eff \sqsubseteq t$---decidable by
Lemma~\ref{lem:effect-order}---is the static enforcement of
look-ahead-freedom. The $\letf$ rule substitutes the bound variable's
effect for $\av(x)$ in the body's effect; because effects lie in $\Av$
and $\Av$ is closed under this substitution, the result stays in $\Av$.

\subsection{Acceptance}
\label{sec:acceptance}

The typing rules carry the temporal content, but soundness for the full
calculus depends on two further well-formedness conditions on the series
operators, surfaced by the hardening of
Section~\ref{sec:series-hardening}. We fold all requirements into a
single acceptance predicate so that the soundness theorem stands with no
external side-conditions.

\begin{definition}[Acceptance]
\label{def:acceptance}
$\vdash_{\Frag} e$ holds if and only if all of the following hold:
\begin{enumerate}
  \item \textbf{(Fragment)} $e\in\Frag$ (Definition~\ref{def:fragment}):
        every availability term in $e$ is drawn from the
        value-independent sublanguage $\Av$;
  \item \textbf{(Typing)} $\varnothing \vdash e : \tau\,!\,\eff$ is
        derivable for some $\tau,\eff$, with every $\decidef_t$ subterm
        discharging its \rulename{T-Decide} premise $\eff\sqsubseteq t$;
  \item \textbf{(Monotone leaves)} every base series in $e$ is declared
        availability-monotone (Definition~\ref{def:monotone}); this is a
        data-source property asserted at ingestion and trusted per the
        stated boundary (incorrect source stamps are out of scope);
  \item \textbf{(Causal alignment)} every $\resamplef_\rho$ in $e$ has a
        causal rule ($\max\rho^{\leftarrow}(p)\preceq\iota(p)$,
        Section~\ref{sec:series-hardening}), and every $\joinf$ aligns
        its sources at a common position rather than across a look-ahead
        offset.
\end{enumerate}
Conditions (3)--(4) are local: each is $O(1)$ per operator, and, by
Lemma~\ref{lem:mono-preserve}, together they establish the global
availability-monotonicity invariant on which Lemma~\ref{lem:causality},
and hence soundness, depends.
\end{definition}

With acceptance so defined, everything the soundness proof requires is a
clause of Definition~\ref{def:acceptance}. We can already discharge
decidability.

\begin{proof}[Proof of Theorem~\ref{thm:decidable-check}]
Type inference is syntax-directed: each rule computes its conclusion's
effect as a $\max$ or substitution of its premises' effects, all in
$\Av$; the only typing side-condition is $\eff\sqsubseteq t$ at
\rulename{T-Decide}, decidable by Lemma~\ref{lem:effect-order}. The
fragment check (clause~1) is the linear syntactic pass of
Proposition~\ref{prop:membership}; the monotone-leaf and causal-alignment
checks (clauses~3--4) are $O(1)$ per operator by inspection of the
declared stamp and alignment arguments. A single bottom-up pass over $e$
therefore computes all effects and discharges all obligations in time
linear in $|e|$; effect normal forms have size bounded by the term, and
arities are fixed.
\end{proof}

\subsection{Soundness via a Two-Run Logical Relation}
\label{sec:soundness}

Soundness is a \emph{relational} property: it constrains two executions
that agree up to the decision epoch. A single-run preservation argument
cannot establish it. We define a family of relations indexed by the
decision epoch $t$ and prove a Fundamental Lemma from which
Theorem~\ref{thm:soundness} follows as the instance at the decision
sort. Throughout, $I \approx_t I'$ denotes that two input assignments
agree on every base datum whose availability is $\le t$
(Definition~\ref{def:agree}).

\begin{definition}[Indistinguishability at $t$]
\label{def:logrel}
Fix an epoch $t$ and inputs $I \approx_t I'$. For results carrying
effect $\eff$, define:
\begin{itemize}
  \item \emph{Values.} $v \sim^{\eff}_{t,I,I'} v'$ iff
    $\dom{\eff}_{t,I}=\dom{\eff}_{t,I'} \le t \Rightarrow v = v'$: when
    the effect certifies that the result is drawn from the admissible
    past, the two runs must agree exactly; when the effect exceeds $t$,
    no agreement is demanded, and by \rulename{T-Decide} such a value
    can never reach a decision.
  \item \emph{Time-indexed values.} $\langle v,\alpha\rangle
    \sim^{\eff}_{t,I,I'} \langle v',\alpha'\rangle$ iff
    $\alpha=\alpha'$ and $v \sim^{\eff}_{t,I,I'} v'$: stamps always
    agree, by the determinacy lemma below; values agree when admissible.
\end{itemize}
\end{definition}

\begin{lemma}[Availability determinacy]
\label{lem:avail-det-final}
If $\Gamma \vdash e : \TVal\,!\,\eff$ in $\Frag$ and $I\approx_t I'$,
and $I\vdash_t e\Downarrow\langle v,\alpha\rangle$ and
$I'\vdash_t e\Downarrow\langle v',\alpha'\rangle$, then $\alpha=\alpha'$.
\end{lemma}
\begin{proof}
Every stamp produced in $\Frag$ is the denotation of an availability
term $A\in\Av$ (by the grammar and rules \rulename{T-Base},
\rulename{T-Stamp}, \rulename{T-Asof}, \rulename{T-Join}). By
Lemma~\ref{lem:value-indep}, $A$ depends on inputs only through the
stamps of referenced data. Each referenced stamp is a constant, $t$, or
$\av(x)$; by induction on the typing derivation, every such $\av(x)$ is
either pinned by $\approx_t$ (when $\le t$) or enters only through
$\max$, $\min$, or $+\delta$ without affecting the value branch. The
lemma is not circular: it does not presuppose acceptance of any
decision, only that stamps are $\Av$-terms, which is a syntactic fact
about $\Frag$. Hence $\alpha=\alpha'$.
\end{proof}

Stamp agreement is thus a consequence of grammar---$\Av$ is
value-independent---and holds independently of the effect discipline;
the effect discipline is needed only for \emph{value} agreement.

\begin{lemma}[Fundamental Lemma]
\label{lem:fundamental}
Let $\Gamma \vdash e : \tau\,!\,\eff$ in $\Frag$. For all $t$ and all
$I \approx_t I'$ with environments related pointwise by
$\sim_{t,I,I'}$ at their declared effects, if $I\vdash_t e\Downarrow r$
and $I'\vdash_t e\Downarrow r'$, then $r \sim^{\eff}_{t,I,I'} r'$.
\end{lemma}
\begin{proof}
By induction on the typing derivation. Stamps agree throughout by
Lemma~\ref{lem:avail-det-final}, so we track only value agreement, and
only in the case $\dom{\eff}_{t,I}\le t$.

\emph{\rulename{T-Base}.} $\langle d,A\rangle$ has effect $A$. If
$\dom{A}_t \le t$ the datum is admissible, so by $\approx_t$ its value
is pinned and $d=d$.

\emph{\rulename{T-Val}, \rulename{T-Op}.} The effect is the $\max$ of
the premises' effects. If $\max_i\eff_i \le t$ then every $\eff_i \le
t$, so by the induction hypothesis each operand value agrees; $f$ is
pure, so $f(\vec v)=f(\vec v')$.

\emph{\rulename{T-Stamp}.} The effect is $\eff$, unchanged by
re-stamping. If $\eff\le t$ the induction hypothesis gives value
agreement, and the new stamp $A$ agrees by
Lemma~\ref{lem:avail-det-final}. The advertised stamp $A$ is irrelevant
to the value relation---this is exactly why re-stamping cannot launder a
future value past the discipline. In the classical taxonomy this is an \emph{intransitive} flow~\cite{rushby1992noninterference}: a restamp is a permitted relabelling channel, but permitting it must not thereby permit the future value it carries to reach the decision.

\emph{\rulename{T-Asof} (and \rulename{T-Retrieve}).} The effect is $A$.
If $A\le t$, the as-of semantics read only store entries with
availability $\le A \le t$, all pinned by $\approx_t$; hence identical
content is selected and the values agree.

\emph{\rulename{T-Window}, \rulename{T-Scan}, \rulename{T-Join}.} The
effects are $\max$es of premises, and the causal semantics reach no
position later than the current one (Lemma~\ref{lem:causality}). If the
effect is $\le t$, all contributing elements are admissible and the
induction hypothesis applies elementwise. The series cases are proved in
full in Section~\ref{sec:series-hardening}, which supplies the causality
and monotonicity results this step invokes.

\emph{\rulename{T-Let}.} By the induction hypothesis $e_1$'s result
relates at $\eff_1$; extend the environment relation at $x$ and apply the
induction hypothesis to $e_2$; the substituted effect
$\eff_2[\eff_1/\av(x)]$ tracks the dependency.

\emph{\rulename{T-Decide}.} The premise $\eff\sqsubseteq t$ forces the
body's effect $\le t$ at this epoch. By the induction hypothesis the
body's value agrees, and its stamp agrees, so the emitted decision
$\langle v,\alpha\rangle$ is identical across the two runs.
\end{proof}

\begin{proof}[Proof of Theorem~\ref{thm:soundness}]
Suppose $\vdash_{\Frag} e$ and fix any $t$, any $I$, and any $I'\approx_t
I$. Every $\decidef_t$ subterm is well-typed with its \rulename{T-Decide}
premise discharged. Applying Lemma~\ref{lem:fundamental} at that subterm
yields an identical emitted $\langle v,\alpha\rangle$ under $I$ and $I'$.
Since $t$, $I$, and $I'$ were arbitrary, $e$ satisfies
Definition~\ref{def:laf}: it is look-ahead-free.
\end{proof}

\subsection{Relative Completeness}
\label{sec:completeness}

The discipline is sound but conservative: it rejects some
look-ahead-free terms whose freedom is semantic rather than
availability-structural---for instance, a value read from the future but
provably cancelled inside a value operation. The following proposition
states the strongest positive result available to any check that treats
value operations as opaque.

\begin{proposition}[Completeness relative to the effect abstraction]
\label{prop:rel-complete}
If $e\in\Frag$ is rejected by $\vdash_{\Frag}$, then there exists an
interpretation of $e$'s value operations, consistent with
Assumption~\ref{ass:pure}, under which $e$ is not look-ahead-free.
Equivalently, the effect system is the most permissive sound discipline
that treats value operations as opaque.
\end{proposition}
\begin{proof}
Rejection means some $\decidef_t$ has body effect
$\eff\not\sqsubseteq t$: a base datum with availability $>t$ whose value
reaches the decision through opaque operations. Choose those operations
to be projections that expose the datum; then a future perturbation of
it changes the decision, violating Definition~\ref{def:laf}. Hence no
sound discipline that cannot inspect inside value operations can accept
$e$.
\end{proof}

This is the honest ceiling of the approach. Exact look-ahead-freedom is
value-semantic, and any decidable, function-agnostic check is complete
only relative to opaque value operations. Inspecting inside a value
operation---recognising, say, the algebraic cancellation of a future
term---recovers more programs at the cost of decidability, which is
another face of the boundary that Theorem~\ref{thm:undecidable} draws.

\section{Hardening the Series Operators}
\label{sec:series-hardening}

The soundness proof of Theorem~\ref{thm:soundness} treated $\scanf$,
$\windowf$, $\resamplef$, and $\joinf$ by appealing to the fact that
their causal semantics reach no position later than the current one. We
now discharge that appeal. This section makes series and their
evaluation precise, defines the availability-monotonicity invariant and
proves that every operator preserves it, establishes a Causality Lemma
certifying that the typed effect of a series term equals the true stamp
of its output, and, on that basis, replaces the asserted series cases of
the Fundamental Lemma with proved ones. It also surfaces the two local
obligations---clauses (3) and (4) of Definition~\ref{def:acceptance}---%
that the checker must discharge.

\subsection{Series and Positions}
\label{sec:series}

A term of series type denotes a function $\bar{e} : \Pos \to \TVal$ from
a totally ordered \emph{position} domain $(\Pos,\preceq)$---bar indices,
event ordinals---to time-indexed values. Positions and times are
distinct: $\Pos$ orders \emph{where} an element sits in a series, while
availability $\av$ records \emph{when} it becomes knowable. We write
$\bar{e}[p] = \langle v_p, \alpha_p\rangle$ for the element at position
$p$, and $\av(\bar e[p]) = \alpha_p$.

\begin{definition}[Availability-monotone series]
\label{def:monotone}
A series $\bar e$ is \emph{availability-monotone} if and only if
$p \preceq q \Rightarrow \av(\bar e[p]) \le \av(\bar e[q])$. We write
$\Mono(\bar e)$.
\end{definition}

Availability-monotonicity says that later positions are never available
earlier---the defining regularity of an honestly stamped time series. It
is a well-formedness invariant, not an assumption about arbitrary data:
base series are monotone by construction, since a source publishes bar
$t$ no earlier than bar $t-1$, and we prove that every operator preserves
it. Where an operator \emph{could} break it---a misaligned $\resamplef$
or $\joinf$---the preservation lemma's premises identify exactly the
well-formedness side-conditions the checker must enforce, turning a
silent leakage channel into an explicit, checked obligation.

\subsection{Operational Semantics of the Series Operators}
\label{sec:series-opsem}

Evaluation of a series term yields, at each position $p$, a time-indexed
value, written $I \vdash \bar e[p] \Downarrow \langle v,\alpha\rangle$.
The rules below are the causal definitions; each references only
positions $\preceq p$.

Every series has a least position $p_0$ under $\preceq$, its \emph{start};
$\predf(p)$ denotes the immediate predecessor of $p$, defined for
$p\ne p_0$. Clamped back-subtraction is
\[
  p \ominus j \;=\;
  \begin{cases}
    \text{the position } j \text{ steps back from } p, & \text{if it is } \succeq p_0,\\[2pt]
    p_0, & \text{otherwise (clamp at the start).}
  \end{cases}
\]
Thus for $p$ within $k-1$ steps of the start, the window
$[\,p\ominus(k{-}1),\,p\,]$ is the truncated prefix $[\,p_0,\,p\,]$
rather than reaching before $p_0$: there is no position earlier than
$p_0$ to read. This fixes the base case of the inductions in
Section~\ref{sec:causality}, where at $p_0$ the window and the scan read
only $p_0$ itself.

The last-$k$ window collects positions in $[p\ominus(k{-}1),p]$, never
beyond $p$, and stamps its output with the latest availability among
them:
\[
\frac{\{\,I \vdash \bar e[q]\Downarrow\langle v_q,\alpha_q\rangle\,\}_{\,p\ominus (k-1)\preceq q\preceq p}}
     {I \vdash \windowf_k(\bar e)[p] \Downarrow
        \big\langle\, [\,v_{q}\,]_{q},\;\; \textstyle\max_{p\ominus(k-1)\preceq q\preceq p} \alpha_q \,\big\rangle}
\quad(\rulename{E-Window})
\]
The causal fold is defined from its predecessor accumulator and the input
at the current position only, and stamps with the running $\max$ of input
stamps:
\[
\frac{I \vdash \bar e[p_0]\Downarrow\langle v_0',\alpha_0'\rangle}
     {I \vdash \scanf(f,e_0,\bar e)[p_0] \Downarrow
        \langle\, f(\val(e_0), v_0'),\; \max(\alpha_{e_0}, \alpha_0') \,\rangle}
\quad(\rulename{E-Scan-Init})
\]
\[
\frac{I \vdash \scanf(f,e_0,\bar e)[\predf(p)]\Downarrow\langle a,\alpha_a\rangle
      \qquad
      I \vdash \bar e[p]\Downarrow\langle v_p,\alpha_p\rangle}
     {I \vdash \scanf(f,e_0,\bar e)[p] \Downarrow
        \langle\, f(a, v_p),\;\; \max(\alpha_a,\alpha_p)\,\rangle}
\quad(\rulename{E-Scan-Step})
\]
This is the rule whose stamp bound the typing rule \rulename{T-Scan} must
match. Join reads position $p$ of each source and stamps with the latest,
formalising the principle that a joined row is available when its last
component is:
\[
\frac{\{\,I \vdash \bar e_i[p]\Downarrow\langle v_i,\alpha_i\rangle\,\}_{i=1}^{n}}
     {I \vdash \joinf(\bar e_1,\dots,\bar e_n)[p] \Downarrow
        \langle\, (v_1,\dots,v_n),\;\; \textstyle\max_i \alpha_i \,\rangle}
\quad(\rulename{E-Join})
\]
Finally, $\resamplef_\rho$ maps an output position $p$ to a set
$\rho^{\leftarrow}(p) \subseteq \Pos$ of source positions aggregated into
it---for example, the minutes composing daily bar $p$:
\[
\frac{\{\,I \vdash \bar e[q]\Downarrow\langle v_q,\alpha_q\rangle\,\}_{q\in\rho^{\leftarrow}(p)}}
     {I \vdash \resamplef_\rho(\bar e)[p] \Downarrow
        \big\langle\, \aggf([v_q]),\;\; \textstyle\max_{q\in\rho^{\leftarrow}(p)} \alpha_q \,\big\rangle}
\quad(\rulename{E-Resample})
\]
The rule $\rho$ is \emph{causal} if $\max \rho^{\leftarrow}(p) \preceq
\iota(p)$, where $\iota(p)$ is the source position identified with output
$p$'s emission point; that is, an output bar aggregates only source
positions at or before its own close. Up-sampling requires the dual
(each output maps to a unique source position $\preceq$ it). A non-causal
$\rho$ is the classic resample leak, which the checker rejects
(clause~(4) of Definition~\ref{def:acceptance}).

\subsection{Preservation of Monotonicity}
\label{sec:preservation}

\begin{lemma}[Monotonicity preservation]
\label{lem:mono-preserve}
If every base series is availability-monotone, then:
\begin{enumerate}
  \item $\windowf_k(\bar e)$ and $\scanf(f,e_0,\bar e)$ are monotone
        whenever $\bar e$ is;
  \item $\joinf(\bar e_1,\dots,\bar e_n)$ is monotone whenever each
        $\bar e_i$ is;
  \item $\resamplef_\rho(\bar e)$ is monotone whenever $\bar e$ is
        monotone and $\rho$ is causal.
\end{enumerate}
\end{lemma}
\begin{proof}
Each output stamp is a $\max$ of input stamps over a position set that is
itself monotone in $p$. For \rulename{E-Window}, $p\preceq p'$ gives
$[p\ominus(k{-}1),p]$ pointwise $\preceq [p'\ominus(k{-}1),p']$, and the
$\max$ over a $\le$-dominated multiset of a monotone series is
non-decreasing. For \rulename{E-Scan-Step}, the stamp at $p$ is
$\max(\text{stamp at }\predf(p),\alpha_p)\ge$ the stamp at $\predf(p)$,
immediately. For \rulename{E-Join}, the componentwise $\max$ of monotone
series is monotone. For \rulename{E-Resample} with causal $\rho$,
$p\preceq p'$ implies $\rho^{\leftarrow}(p)$ lies pointwise $\preceq
\rho^{\leftarrow}(p')$ and each is $\preceq$ its emission point, so the
$\max$ is non-decreasing.
\end{proof}

Monotonicity is thus an invariant of every well-formed $\Frag$-pipeline,
established once at the leaves and preserved at every node---exactly the
discipline that lets the effect terms of $\Av$, which use $\max$ and
$\min$ over positions, soundly bound availability.

\subsection{The Causality Lemma}
\label{sec:causality}

\begin{lemma}[Causality]
\label{lem:causality}
Let $\bar e$ be a series term of $\Frag$ built from monotone base series
by $\windowf$, $\scanf$, $\joinf$, and causal $\resamplef$. Then for
every position $p$:
\begin{enumerate}
  \item \textup{(Position causality)} the value and stamp of $\bar e[p]$
        are functions only of base-series elements at positions
        $\preceq_{\!*} p$, where $\preceq_{\!*}$ is the pointwise-back
        reach of the operators (the identity for window, scan, and join;
        $\rho^{\leftarrow}$ for resample); and
  \item \textup{(Stamp equals effect bound)} the stamp $\alpha_p$ equals
        the $\max$ availability over those reachable base elements, and
        this equals $\dom{\eff}$ for the effect $\eff$ assigned by the
        typing rules \rulename{T-Window}, \rulename{T-Scan},
        \rulename{T-Join}, and \rulename{T-Resample}.
\end{enumerate}
\end{lemma}
\begin{proof}
By induction on the series-term structure, reading off the operational
rules of Section~\ref{sec:series-opsem}.

\emph{Base series.} $\bar e[p]$ is the datum at $p$; its reach is
$\{p\}$; its stamp is $\alpha_p$; and its effect (\rulename{T-Base}) is
that stamp.

\emph{Window.} \rulename{E-Window} reads positions
$[p\ominus(k{-}1),p]\subseteq \{q : q\preceq p\}$ and stamps with their
$\max$ availability. By the induction hypothesis each is a function of
base elements $\preceq$ itself, hence $\preceq p$; by $\Mono$, the $\max$
over the window equals the availability of position $p$, the latest,
matching \rulename{T-Window}'s effect.

\emph{Scan.} By induction on positions using
\rulename{E-Scan-Init}/\rulename{Step}: the accumulator at $p$ is a
function of inputs at $\{p_0,\dots,p\}$, all $\preceq p$; the stamp is the
$\max$ of those input stamps, which by $\Mono$ is the availability at
$p$, matching $\max(\eff_0,\eff)$ evaluated at $p$. This is the case the
soundness proof previously asserted; it is now a consequence of
\rulename{E-Scan-Step} reading only $\predf(p)$ and $p$.

\emph{Join.} \rulename{E-Join} reads position $p$ of each source; the
reach is the union of the per-source reaches, each $\preceq p$ by the
induction hypothesis; the stamp is $\max_i\alpha_i$, matching
\rulename{T-Join}.

\emph{Resample (causal $\rho$).} \rulename{E-Resample} reads
$\rho^{\leftarrow}(p)$, all $\preceq \iota(p)$ by causality, and stamps
with their $\max$; by the induction hypothesis and $\Mono$ this is
bounded by the availability at $\iota(p)$, matching
\rulename{T-Resample}.
\end{proof}

The Causality Lemma is precisely the ``reaches no position later than the
current one'' claim that the Fundamental Lemma required, now proved, and
it additionally certifies that the typed effect equals the true stamp,
closing any gap between the static bound and runtime reality.

\subsection{The Series Cases of the Fundamental Lemma}
\label{sec:fundamental-rewrite}

We can now discharge the series cases of Lemma~\ref{lem:fundamental}
without appeal to unproven behaviour. Fix $t$ and $I\approx_t I'$, and
consider a series subterm evaluated toward a decision at $t$ whose effect
$\eff$ satisfies $\dom{\eff}_{t}\le t$; otherwise no value agreement is
demanded, and by \rulename{T-Decide} the subterm cannot reach the
decision.

\emph{Window, scan, join, resample.} By the Causality Lemma
(Lemma~\ref{lem:causality}), the value at the relevant position $p$ is a
function of base elements at positions whose availability is
$\le \dom{\eff}_t \le t$. Every such base element therefore has $\av\le t$
and is pinned by $\approx_t$: identical value and stamp under $I$ and
$I'$. The operator's value function---a composition of pure operations
and selection---applied to identical inputs yields identical output, and
the stamp is identical by Lemma~\ref{lem:avail-det-final}. Hence
$\bar e[p]$ agrees across the two runs.

Every step now cites either the Causality Lemma (reach $\preceq p$ with
stamp equal to effect), monotonicity preservation (the invariant that
makes the effect bound tight), or $\approx_t$-pinning (admissible data
agrees), so the series cases stand on the same footing as the scalar
ones.

\subsection{The Resulting Checker Obligations}
\label{sec:series-check}

The hardening surfaces exactly the two obligations recorded as
clauses~(3) and (4) of Definition~\ref{def:acceptance}, both syntactic
and local. First, base series must be declared availability-monotone---a
property of the data source, asserted at ingestion and, per the stated
boundary, trusted, since incorrect source stamps remain out of scope.
Second, each $\resamplef_\rho$ must carry a causal $\rho$
($\max\rho^{\leftarrow}(p)\preceq\iota(p)$), and each $\joinf$ must align
its sources at a common position rather than across a look-ahead offset;
both are checkable by inspecting the operator's alignment argument, and a
non-causal alignment is rejected. With these local checks,
Lemma~\ref{lem:mono-preserve} guarantees the global monotonicity
invariant and Lemma~\ref{lem:causality} guarantees effect-stamp
agreement, so Theorem~\ref{thm:soundness} goes through as written.
Decidability (Theorem~\ref{thm:decidable-check}) is unaffected: both
obligations are $O(1)$ per operator, preserving the linear-time bound.

\section{Operational Semantics and Look-Ahead-Freedom}
\label{sec:semantics}

This section fixes the dynamic meaning of a pipeline and, on that basis,
states the property the discipline of Section~\ref{sec:typesystem}
enforces. We define an epoch-parametrised big-step evaluation, the
notion of two inputs \emph{agreeing up to an epoch}, and finally
\emph{look-ahead-freedom} as a two-run temporal non-interference
property. The definitions here are what Sections~\ref{sec:calculus}%
--\ref{sec:series-hardening} refer forward to; in the final assembly
this section may precede the calculus, so that the property is stated
before the machinery that checks it.

\subsection{Inputs, Stores, and Admissibility}
\label{sec:sem-inputs}

An \emph{input assignment} $I$ interprets the base data of a pipeline:
to each base datum $\langle d, A\rangle$ occurring in a term it assigns a
time-indexed value $I(\langle d,A\rangle) = \langle v,\alpha\rangle$,
where $\alpha = \dom{A}$ is the datum's availability stamp and $v$ its
value. For the vintage constructs, $I$ additionally provides a
\emph{store}: a set of time-indexed entries, each carrying its own
availability, over which $\asoff$ and $\retrievef$ read. We write
$\dom{A}_{t,I}$ for the denotation of an availability term $A$ at epoch
$t$ under $I$; by Lemma~\ref{lem:value-indep} this depends on $I$ only
through the stamps of referenced data.

A time-indexed value $\langle v,\alpha\rangle$ is \emph{admissible at
epoch $t$} if $\alpha \le t$: its value is knowable by $t$. A decision
made at $t$ may legitimately depend only on admissible data. This is the
semantic content that the effect discipline approximates statically.

\subsection{Epoch-Parametrised Evaluation}
\label{sec:sem-eval}

Evaluation is a big-step relation
\[
  I \vdash_t e \Downarrow r,
\]
read: under input assignment $I$, evaluating term $e$ toward a decision
at epoch $t$ produces result $r$ (a value, a time-indexed value, or a
series thereof). The epoch $t$ is a parameter because the calculus's
decision construct, and the as-of reads it governs, are relative to the
epoch at which a decision is being formed. The rules for the series
operators were given in Section~\ref{sec:series-opsem}
(\rulename{E-Window}, \rulename{E-Scan-Init/Step}, \rulename{E-Join},
\rulename{E-Resample}); we give here the remaining, scalar rules.

A base datum evaluates to the value and stamp the assignment gives it; a
projection discards the stamp; a pure operation applies once its
operands are evaluated:
\[
\frac{I(\langle d,A\rangle) = \langle v,\alpha\rangle}
     {I \vdash_t \langle d,A\rangle \Downarrow \langle v,\alpha\rangle}
\quad(\rulename{E-Base})
\qquad
\frac{I \vdash_t e \Downarrow \langle v,\alpha\rangle}
     {I \vdash_t \val(e) \Downarrow v}
\quad(\rulename{E-Val})
\]
\[
\frac{\{\,I \vdash_t e_i \Downarrow v_i\,\}_{i=1}^{n}}
     {I \vdash_t f(e_1,\dots,e_n) \Downarrow f(v_1,\dots,v_n)}
\quad(\rulename{E-Op})
\]
Re-stamping evaluates its subject and replaces the stamp with the
denotation of $A$, leaving the value untouched---the operational
counterpart of \rulename{T-Stamp}'s ``the value still depends on what it
depended on'':
\[
\frac{I \vdash_t e \Downarrow v \qquad \dom{A}_{t,I} = \alpha}
     {I \vdash_t \stampf(e,A) \Downarrow \langle v,\alpha\rangle}
\quad(\rulename{E-Stamp})
\]
A point-in-time read returns the store content selected as of the
denotation of $A$: among store entries with availability $\le
\dom{A}_{t,I}$, the read returns the applicable one (e.g.\ the latest
such entry for the queried key), stamped with $\dom{A}_{t,I}$.
$\retrievef$ evaluates identically over the agentic context store:
\[
\frac{\dom{A}_{t,I} = \alpha \qquad
      w = \mathsf{read}\big(\mathrm{store}(I),\, \{\text{entries with } \av \le \alpha\}\big)}
     {I \vdash_t \asoff(e_{\text{store}}, A) \Downarrow \langle w,\alpha\rangle}
\quad(\rulename{E-Asof})
\]
A decision at epoch $t$ evaluates its body \emph{at that same epoch} and
emits the result:
\[
\frac{I \vdash_t e \Downarrow \langle v,\alpha\rangle}
     {I \vdash_t \decidef_t(e) \Downarrow \langle v,\alpha\rangle}
\quad(\rulename{E-Decide})
\]
and $\letf$ evaluates its bound term, extends the assignment for the
bound variable, and evaluates the body:
\[
\frac{I \vdash_t e_1 \Downarrow r_1 \qquad
      I[x \mapsto r_1] \vdash_t e_2 \Downarrow r_2}
     {I \vdash_t \letf\ x = e_1\ \inkw\ e_2 \Downarrow r_2}
\quad(\rulename{E-Let})
\]
The reference-time versus availability distinction is now operational:
the \emph{position} at which a value sits in a series, and the reference
period a datum is \emph{about}, are recorded in its value/position; the
\emph{availability} $\alpha$ records only when it becomes knowable, and
it is $\alpha$---not the reference time---that admissibility and the
whole discipline turn on. A datum about period $t{-}1$ that becomes
knowable only at $t{+}5$ has availability $t{+}5$, and using it in a
decision at $t$ is a look-ahead violation despite its earlier reference
time. This is the confusion the calculus is designed to make
impossible to express accidentally within the fragment.

\subsection{Agreement Up to an Epoch}
\label{sec:sem-agree}

Look-ahead-freedom compares two executions that are identical on the
admissible past and may differ arbitrarily on the future. We name that
relation on inputs.

\begin{definition}[Agreement up to $t$]
\label{def:agree}
Two input assignments $I$ and $I'$ \emph{agree up to epoch $t$}, written
$I \approx_t I'$, if they assign identical time-indexed values to every
base datum, and identical entries to every store entry, whose
availability is $\le t$:
\[
  I \approx_t I' \quad\stackrel{\text{def}}{\iff}\quad
  \forall\, \text{base data and store entries } b:\;
  \av(b) \le t \;\Rightarrow\; I(b) = I'(b).
\]
No constraint is placed on data with availability $> t$: the two
assignments may differ freely on the future relative to $t$.
\end{definition}
$\approx_t$ is an equivalence relation for each fixed $t$, and it is
monotone in the epoch: $t' \le t$ and $I \approx_t I'$ imply
$I \approx_{t'} I'$ (agreeing on a larger admissible set implies
agreeing on a smaller one). When we say a datum is \emph{pinned} by
$\approx_t$, we mean its availability is $\le t$, so
Definition~\ref{def:agree} forces the two runs to assign it the same
time-indexed value---the fact used pervasively in the soundness proof.

\subsection{Look-Ahead-Freedom}
\label{sec:sem-laf}

We can now state the property. Informally: a pipeline is look-ahead-free
if no decision it emits at any epoch $t$ can be changed by altering data
that only becomes available after $t$. Formally, it is a two-run
non-interference statement over the time-indexed partition induced by
$t$: the ``high'' (protected) inputs are the data with availability
$> t$, the ``low'' (observable) output is the decision at $t$, and
non-interference is the requirement that the low output is a function of
the low inputs alone.

\begin{definition}[Look-ahead-freedom]
\label{def:laf}
A pipeline term $e$ is \emph{look-ahead-free} if for every epoch $t$,
every $\decidef_t$ subterm of $e$ with body $e'$, and every pair of
input assignments $I \approx_t I'$,
\[
  I \vdash_t \decidef_t(e') \Downarrow \langle v,\alpha\rangle
  \ \text{ and }\
  I' \vdash_t \decidef_t(e') \Downarrow \langle v',\alpha'\rangle
  \ \Longrightarrow\ v = v'.
\]
That is, the \emph{value} emitted by the decision at $t$ is invariant
under any perturbation of the data whose availability exceeds $t$.
\end{definition}
The property is stated on the decision's emitted \emph{value}, since the
value is the observable a downstream consumer acts on and the only thing
a future datum could corrupt; the decision's stamp is $t$ by
\rulename{E-Decide} and so agrees across the two runs unconditionally.
This is the standard non-interference posture---the observable low output
is a function of the low inputs alone---and it is the exact property the
two-run logical relation of Section~\ref{sec:soundness} was built to
establish, where value agreement at effect $\eff\le t$ is precisely
$v=v'$ here.

This is the property Theorem~\ref{thm:soundness} certifies for accepted
terms of $\Frag$, and the property Theorem~\ref{thm:undecidable}
(Section~\ref{sec:decidability}) shows is undecidable on the full
language $\Full$.

\section{The Undecidability Boundary}
\label{sec:decidability}

Theorem~\ref{thm:soundness} gives a sound, decidable check on the
fragment $\Frag$. We now show this is the best possible in a precise
sense: on the full language $\Full$, where availability may be computed
from data values (Definition~\ref{def:full}), look-ahead-freedom is
undecidable. The fragment restriction is therefore not a convenience but
the exact frontier of decidability.

\begin{theorem}[Undecidability on the full language]
\label{thm:undecidable}
Look-ahead-freedom (Definition~\ref{def:laf}) is undecidable for terms of
$\Full$. More precisely, the set of look-ahead-free terms of $\Full$ is
$\Pi^0_1$-complete: non-look-ahead-freedom is recursively enumerable---a
leak is a finite witness---while look-ahead-freedom is not.
\end{theorem}

The proof is a many-one reduction from the halting problem. We first fix
the modelling assumptions the encoding uses; each concerns $\Full$, the
object whose undecidability we characterise, not the fragment $\Frag$
where real pipelines live.

\begin{description}
  \item[\textnormal{(M1)}] \emph{Value-conditioned availability.} $\Full$
    admits, as the availability argument of a stamping construct, a
    pipeline term of sort $\Time$ computed from data values---in
    particular a total conditional such as
    $\mathsf{if}\ \val(e){>}c\ \mathsf{then}\ \tau_1\ \mathsf{else}\
    \tau_2$. This is exactly the extension defining $\Full$
    (Definition~\ref{def:full}), whose own illustrative instance is a
    value-conditioned $\Time$ term.
  \item[\textnormal{(M2)}] \emph{Configurations as values.} A value
    ($\Val$) may hold a finite string, so a Turing-machine
    configuration---control state, head position, and finite tape---is a
    $\Val$, and one step of the machine is a total, pure function
    $\Val\to\Val$, consistent with Assumption~\ref{ass:pure}. The step
    function is total; we never ask a value operation to decide halting,
    only to advance one step.
  \item[\textnormal{(M3)}] \emph{Countable position domains.} Series in
    $\Full$ may be countably infinite, with position domain an initial
    segment of $\mathbb{N}$ of unbounded length. Because every operator
    is causal and pointwise (Section~\ref{sec:series-opsem})---$\scanf$
    depends only on positions $\preceq p$, $\windowf$ on a bounded
    suffix, and we use no unbounded-lookback resample---each element
    $\bar e[p]$ is finitely determined even when the series is infinite.
    No evaluation ever forces a completed infinite series; the infinity
    resides solely in the epoch quantifier of
    Definition~\ref{def:laf}.
\end{description}

\begin{proof}
Let $M$ be a Turing machine (run on blank input; the general case is
identical). We construct a term $e_M\in\Full$ with
\[
  M \text{ halts} \quad\Longleftrightarrow\quad
  e_M \text{ is \emph{not} look-ahead-free},
\]
which reduces the halting problem to non-look-ahead-freedom and hence
establishes undecidability.

\emph{Simulating $M$ with a causal fold.} Identify the position domain
with the epoch grid: position $p$ corresponds to epoch $p$, and data at
position $p$ has availability $p$. By (M2), let $\mathit{cfg}_0$ be the
initial configuration and $\mathsf{step}:\Val\to\Val$ the total one-step
function. The causal fold
\[
  \bar c \;=\; \scanf(\mathsf{step},\, \mathit{cfg}_0,\, \bar\iota)
\]
over a driver series $\bar\iota$ (a base series with one element per
position, availability $p$ at position $p$) computes, by
\rulename{E-Scan-Init}/\rulename{E-Scan-Step}, the configuration
$\mathit{cfg}_p = \mathsf{step}^{p}(\mathit{cfg}_0)$ at each position
$p$. By causality (Lemma~\ref{lem:causality}), $\bar c[p]$ depends only on
positions $\preceq p$, so $\mathit{cfg}_p$ is a function of data available
by epoch $p$; its availability stamp is $p$. Thus $\mathit{cfg}_p$ is
\emph{admissible at epoch $p$}.

\emph{A halting predicate over admissible data.} Let
$\mathsf{halted}:\Val\to\mathsf{Bool}$ be the total pure predicate true
exactly on halting configurations (M2). The value
$h_t = \mathsf{halted}(\mathit{cfg}_t)$ is computed from $\mathit{cfg}_t$,
which is admissible at epoch $t$; hence $h_t$ is itself admissible at
epoch $t$---it reads no future data. Note $h_t$ is monotone in $t$ once
true (a halting machine stays halted), and $h_t=\mathit{false}$ for all
$t$ iff $M$ never halts.

\emph{Gating a genuinely future datum by the halting predicate.}
Introduce a base datum $d^{\uparrow}$ with a genuinely future
availability---formally, at epoch $t$ its natural availability is
$t{+}1>t$, so it is inadmissible at $t$. Using the value-conditioned
availability of (M1), re-stamp $d^{\uparrow}$ with
\[
  A_t \;=\; \mathsf{if}\ h_t\ \mathsf{then}\ t\ \mathsf{else}\ t{+}1,
  \qquad
  d^{\star} = \stampf(d^{\uparrow}, A_t),
\]
and feed its value into the decision at $t$:
\[
  e_M \;=\; \decidef_t\big(\, g(\val(d^{\star}))\,\big),
\]
where $g$ is any value operation with $g$ non-constant in its argument
(so the datum's value genuinely influences the output). The construction
is uniform in $t$: $e_M$ emits a decision at every epoch.

\emph{Correctness of the reduction.} Fix an epoch $t$ and consider two
input assignments $I\approx_t I'$ that differ only on $d^{\uparrow}$
(legitimate, since $d^{\uparrow}$ has availability $>t$ under any
assignment where it is not re-stamped down).

If $M$ halts, let $T$ be its halting step. At epoch $t=T$, $h_T=
\mathit{true}$, so $A_T = T$: the re-stamp marks $d^{\star}$ as available
\emph{now}, its value is admissible, and it flows through $g$ into
$\decidef_T$. Perturbing $d^{\uparrow}$ (a datum with natural
availability $>T$) changes $\val(d^{\star})$ and hence the emitted value,
so there exist $I\approx_T I'$ with differing decisions: $e_M$ is not
look-ahead-free.

If $M$ never halts, then $h_t=\mathit{false}$ for every $t$, so
$A_t=t{+}1>t$ at every epoch: $d^{\star}$ is inadmissible at every
decision, its value never influences any $\decidef_t$, and every decision
depends only on admissible data. Hence for all $t$ and all $I\approx_t
I'$ the emitted values coincide: $e_M$ is look-ahead-free.

Therefore $M$ halts iff $e_M$ is not look-ahead-free, completing the
reduction. Since the halting problem is undecidable, so is
look-ahead-freedom on $\Full$.

\emph{Complexity.} Non-look-ahead-freedom is recursively enumerable: a
leak is witnessed by a finite object---an epoch $t$, and two assignments
$I\approx_t I'$ differing on some datum of availability $>t$ that yields
different decisions at $t$---and each candidate witness is checkable in
finite time by evaluating the two runs up to $t$, which by (M3) touches
only finitely much of the series. Enumerating epochs and finite
perturbations therefore semi-decides non-look-ahead-freedom, placing it
in $\Sigma^0_1$ and look-ahead-freedom in $\Pi^0_1$. The reduction above
maps halting (which is $\Sigma^0_1$-complete) to non-look-ahead-freedom,
so the latter is $\Sigma^0_1$-hard and look-ahead-freedom is
$\Pi^0_1$-hard; with the upper bound, look-ahead-freedom is
$\Pi^0_1$-complete.
\end{proof}

Two consequences frame the contribution. First, the asymmetry is
practically exact: a sound analysis can always \emph{report a leak} it
finds (leaks are r.e.), but no analysis can \emph{certify freedom} in
general (freedom is properly $\Pi^0_1$)---which is precisely why a
differential tester such as DataFlow can flag violations it happens to
exercise yet cannot guarantee their absence, and why the fragment
restriction is the route to a soundness guarantee rather than a
best-effort test. Second, the reduction pinpoints the single feature
responsible: value-conditioned availability (M1). Removing it---confining
availability to the value-independent sublanguage $\Av$---is exactly the
passage from $\Full$ to $\Frag$, and it is what turns an undecidable
property into the linear-time check of Theorem~\ref{thm:decidable-check}.
The boundary is not incidental; it is drawn by one syntactic capability.

\section{Evaluation}
\label{sec:evaluation}

The formal results establish that the discipline is sound and decidable
in linear time. This section asks three empirical questions that the
theory alone does not answer. First, does the checker's cost actually
scale linearly in practice, as Theorem~\ref{thm:decidable-check}
predicts, and where does the availability-term complexity enter
(\S\ref{sec:eval-scaling})? Second, does soundness---an ACCEPT verdict
never being falsifiable---hold on proprietary market data, as judged by an independent
oracle (\S\ref{sec:eval-oracle})? Third, does the sound checker detect
leaks that the empirical detectors in current practice miss, and at what
false-positive cost (\S\ref{sec:eval-comparison})?

All experiments use a clean-room implementation of the
checker over the calculus of Sections~\ref{sec:calculus}%
--\ref{sec:series-hardening}, with no proprietary code, market
data, or strategy content. The oracle validation
(\S\ref{sec:eval-oracle}) additionally runs against archetypes derived
from proprietary market data that we are not permitted to redistribute;
the reported figures for that experiment therefore come from data readers
cannot obtain. The publicly released artifact reproduces the method, not
those figures: it ships the checker, the detectors, the oracle, the
adversarial corpus, and synthetic stand-ins for the archetypes that
exercise the same leak mechanisms, so that every qualitative claim below
can be re-run end to end (see the Data and Code Availability statement).
Every reported quantity is produced by an
executed experiment against a committed revision; we mark computed
figures accordingly and fix all seeds. We follow a pre-registered
analysis plan committed before the statistical code was written, and
report deviations from it explicitly. Two scope caveats apply
throughout. The comparison corpus (\S\ref{sec:eval-comparison}) is
adversarially constructed to stress detectors, so its rates are
existence and robustness evidence over a hard corpus, not estimates of
real-world leak frequency. The real-data validation
(\S\ref{sec:eval-oracle}) uses an independent dynamic oracle whose
witnesses establish leaks but whose silence does not establish freedom
(freedom is not r.e., by Theorem~\ref{thm:undecidable}).

\subsection{Scaling of the Checker (Claim A)}
\label{sec:eval-scaling}

Theorem~\ref{thm:decidable-check} states that acceptance is decidable in
time linear in term size when availability-term complexity is bounded.
We test this along two axes: the number of pipeline nodes at bounded
availability complexity (Axis~1), and the availability-term complexity
at fixed node count (Axis~2). Timing uses a monotonic high-resolution
clock, five discarded warm-up runs followed by fifty timed repetitions
per size, with pipeline construction excluded from the measured region
and a $5{,}000$-draw bootstrap confidence interval on the median; the
fitted quantity is the ordinary-least-squares slope of $\log(\text{time})$
on $\log(\text{size})$, with a bootstrap interval over resampled
repetitions.

\emph{Axis~1: node count.} Over pipeline sizes from $10$ to $10^5$
nodes at bounded availability complexity, the fitted log--log slope is
$1.023$ (95\% bootstrap CI $[1.021, 1.026]$)
(Figure~\ref{fig:scaling}a). A slope of unity is linear scaling; the
measured slope is within three percent of unity across four orders of
magnitude, confirming the linear-time prediction of
Theorem~\ref{thm:decidable-check}.

We report one methodological point that materially affected this result,
in the interest of reproducibility. An initial implementation of the
availability normal-form computation performed a quadratic
tuple-concatenation fold, yielding a fitted slope of $1.39$---%
super-linear, and at first reading an apparent tension with the theorem.
Profiling localised the cost to the normalisation fold, not the
type-checking pass; replacing the fold with a linear accumulation
(a verdict-preserving change, with the full test suite green before and
after) restored the slope to $1.023$. The theorem predicts linear
behaviour of the \emph{algorithm}; the benchmark first measured an
\emph{implementation} that departed from it, and then, once corrected,
confirmed it. We keep the episode in the record because it is exactly
the kind of gap between a proof about an algorithm and the artifact that
realises it that an empirical scaling check exists to catch.

\emph{Axis~2: availability-term complexity.} Holding node count fixed
and growing the availability term as a balanced alternating
$\min$/$\max$ tree, check time grows smoothly through small depths (from
$1.35\times10^{-5}$\,s at $2$ literals to $5.6\times10^{-4}$\,s at $32$
literals) and then sharply---$4.5\times10^{-2}$\,s at $64$ literals and
$1.1\times10^{-1}$\,s at $128$ literals. This
inflection is not an implementation defect: it is the intrinsic cost of
normalising an alternating $\min$/$\max$ availability term, whose normal
form can be large, and it is the reason
Theorem~\ref{thm:decidable-check} states linearity \emph{under bounded
availability-term complexity} rather than unconditionally
(Figure~\ref{fig:scaling}b). The two axes
together give the honest decomposition: the checker is linear in the
structural size of the pipeline, and the only super-linear regime lies
in adversarially deep availability expressions, which do not arise in
the point-in-time pipelines the fragment is designed for---base data,
lags, and $\max$-of-sources joins produce shallow availability terms.
The complexity of the check is thus governed by how convoluted a
pipeline's \emph{availability structure} is, not by its data-flow size,
and for realistic pipelines the former is small. This observation also
motivates the fragment's restriction to fixed-arity $\min$/$\max$: it is
what keeps availability normal forms, and therefore the effect-ordering
decision, tractable.

\begin{figure}[t]
  \centering
  \includegraphics[width=\linewidth]{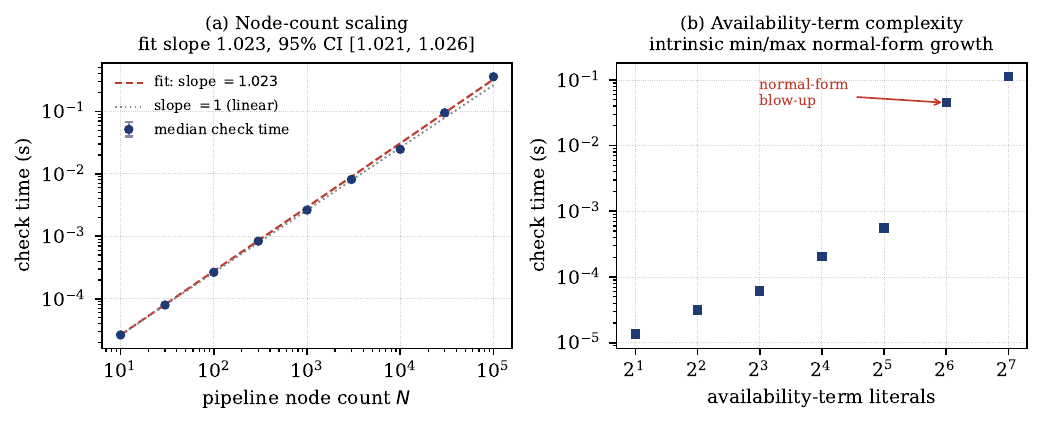}
  \caption{Checker scaling. (a) Node-count
  scaling at bounded availability complexity: median check time versus
  pipeline node count on log--log axes, with a fitted slope of $1.023$
  (95\% bootstrap CI $[1.021, 1.026]$) against the linear
  ($\text{slope}=1$) reference; error bars are $5{,}000$-draw bootstrap
  CIs on the median. (b) Availability-term complexity at fixed node
  count: check time versus the number of literals in a balanced
  alternating $\min$/$\max$ availability term, exhibiting the intrinsic
  normal-form blow-up that the fragment's fixed-arity restriction is
  designed to avoid. Together the panels show the check is linear in
  data-flow size and super-linear only in adversarially deep
  availability structure, which does not arise in point-in-time
  pipelines.}
  \label{fig:scaling}
\end{figure}

\subsection{Soundness on Real Data (Oracle Validation)}
\label{sec:eval-oracle}

Claim~A concerns cost and Claim~B concerns detection against other
tools; neither directly tests the property the theory actually
guarantees---that an \textsc{accept} verdict is never wrong. We test
that here against proprietary market data, using an oracle that decides leak-reality
independently of the checker. The oracle is a dynamic two-run test: at a
sampled epoch $t$, it computes the decision on the data as stored, then
recomputes it after perturbing every datum whose availability exceeds
$t$; if the two decisions differ, the future has influenced the present
and a leak is \emph{witnessed}. This is a direct operationalisation of
Definition~\ref{def:laf} at the sampled epochs, and it shares no code
with the static checker. The soundness question is then precise: does
any pipeline the checker \textsc{accept}s admit an oracle witness? If so,
the checker is unsound; if not, every acceptance withstands an
independent dynamic falsification attempt.

We run the oracle over five archetypes drawn from proprietary market
data (not redistributable), each in a clean and one or more leaky
variants: AR2
(fundamentals, with a reference-time leak and a restatement-vintage
leak), AR3 (weekly releases), AR4 (intraday bars aggregated at a session
close), and AR5 (a survivorship-filtered universe). These span the leak
families that occur in practice---reference-versus-availability
misalignment, restatement vintages, post-close aggregation, and
survivorship. We report only verdicts, witness existence, and decision
deltas; the released artifact substitutes synthetic stand-ins with the
same schemas and leak mechanisms for these archetypes.

Table~\ref{tab:oracle} gives the result. Every
checker \textsc{accept}---AR2, AR3, AR4, and AR5 in their clean
variants---has no oracle witness: the soundness cross-check passes with
no accepted pipeline falsifiable by the oracle. Every leaky variant is
rejected by the checker. For four of the five leaky variants (AR2
reference, AR2 restatement, AR4, AR5) the oracle independently witnesses
the leak, confirming the rejection dynamically; the AR2 restatement case
is notable, as the restatement-vintage leak is one of the subtlest in
practice and the checker catches it through the \rulename{T-Asof} and
\rulename{T-Stamp} discipline.

The AR3 leaky variant is the instructive exception and we call it out
explicitly: the checker rejects it, but the oracle finds no witness at
its sampled epochs. This is not a soundness violation---soundness
constrains \textsc{accept}, and the checker did not accept---but a case
where the sound static analysis catches a real leak (real by the term's
semantics: the alignment reads a value whose availability follows the
decision epoch) that the dynamic oracle fails to exercise, because the
weekly release lag is short and the sampled epochs do not happen to flip
the selected row. It is the same phenomenon as the coverage misses of
Claim~B, now arising on proprietary market data and against a dynamic oracle rather
than a detector: a two-run test certifies nothing by its silence, exactly
as Theorem~\ref{thm:undecidable} predicts, whereas the static checker's
rejection is driven by the term's availability structure regardless of
which epochs a tester samples. The static and dynamic views thus agree
wherever the dynamic view can see, and where they diverge it is the
static, sound analysis that is correct.

\begin{table}[t]
  \caption{Real-data oracle validation. For each
  archetype and variant: the checker verdict, whether the independent
  two-run oracle witnessed a leak, and whether the two agree. No checker
  \textsc{accept} has an oracle witness (soundness holds); the AR3 leaky
  variant is a conservative rejection the oracle does not witness (see
  text). Seed $20260702$; captured at commit \texttt{fcc6a41}.}
  \label{tab:oracle}
  \centering
  \small
  \begin{tabular}{llccl}
    \toprule
    Archetype & Variant & Checker & Oracle witness? & Agreement \\
    \midrule
    AR2 & clean             & \textsc{accept} & no  & yes \\
    AR2 & leaky reference   & \textsc{reject} & yes & yes \\
    AR2 & leaky restatement & \textsc{reject} & yes & yes \\
    AR3 & clean             & \textsc{accept} & no  & yes \\
    AR3 & leaky             & \textsc{reject} & no  & yes (conservative) \\
    AR4 & clean             & \textsc{accept} & no  & yes \\
    AR4 & leaky             & \textsc{reject} & yes & yes \\
    AR5 & clean             & \textsc{accept} & no  & yes \\
    AR5 & leaky             & \textsc{reject} & yes & yes \\
    \bottomrule
  \end{tabular}
\end{table}

\subsection{Detection versus the Empirical Detectors (Claim B)}
\label{sec:eval-comparison}

The practical claim motivating the whole development is that a sound
static check catches leaks that the empirical detectors in current use
miss. We test this on an adversarial corpus of $53$ pipelines with
ground truth fixed by construction: $33$ genuinely leaking (\textsc{leak})
and $20$ genuinely clean (\textsc{clean}). Each leak case carries an
explicit witnessing perturbation---an input and epoch at which the
decision provably changes---so ground truth is established independently
of every tool under test. We compare three tools: the sound checker of
this paper; a \emph{two-run} differential detector that re-executes the
pipeline on perturbed future data at sampled epochs and flags a change;
and a \emph{tiling} detector that compares outputs across shifted data
windows. The two-run detector is the design underlying differential
leakage testers such as DataFlow; the tiling detector represents
window-shift consistency checks.

\emph{The corpus.} The corpus is organised into six mechanisms, each
targeting a specific way a leak evades a detector or a clean pipeline
trips a sound checker. Three produce leaks: \textbf{M1} (coverage miss)
places the leak at an epoch falling between the two-run detector's
sampled epochs, so a detector that samples the time axis can step over
it; \textbf{M2} (perturbation-insensitivity miss) reads a future value
through a threshold operation that the detector's default perturbation
does not flip on the tested input, though a different input would;
\textbf{M3} (insidious re-stamp miss) re-stamps a future-valued series
to an admissible stamp, the empirical form of the \rulename{T-Stamp}
channel---an output-comparing detector sees an unchanging value and
misses it, while the effect system preserves the future effect through
the re-stamp and rejects. Three produce clean pipelines that probe the
checker's precision: \textbf{FP1} (algebraic cancellation) subtracts a
future value from itself; \textbf{FP2} (dominated future) applies an
operation that ignores its future argument; both are clean, and the
checker rejects them because it treats value operations as opaque---the
documented relative-completeness ceiling of
Proposition~\ref{prop:rel-complete}. The sixth, \textbf{FP2b}
(availability-structure domination), declares an availability of
$\min(t, t{+}1)$, which the effect normal form correctly reduces to $t$:
these cases are clean and the checker \emph{accepts} them, testing
whether the checker over-rejects on availability structure. Of the $53$
cases, $52$ are synthetic terms over the public DSL; one is an anchor
drawn from proprietary market data (a weekly-release alignment) included
to exhibit the same detector-miss phenomenon on real market data. In the
released artifact this anchor is replaced by a synthetic term with the
same alignment structure.

\emph{Results.} Table~\ref{tab:detection} reports false-negative rates
over the $33$ leak cases and false-positive rates over the $20$ clean
cases. The sound checker has a false-negative rate
of $0$: it misses none of the $33$ leaks. Because this is a
zero-event arm, we report it with the rule-of-three upper bound, giving
a false-negative rate of $0$ with 95\% upper bound $9.1\%$. The two-run
detector misses $18$ of $33$ leaks ($54.5\%$) and the tiling detector
misses all $33$ ($100\%$); both differences against the checker are
large. On the clean cases, the checker false-positives on $15$ of $20$
($75\%$)---exactly the FP1 and FP2 cases, where value semantics it
cannot inspect render an opaque rejection---while correctly accepting
all five FP2b cases, confirming that the false positives are confined to
opaque value operations and that the effect normal form does not
over-reject on availability structure. The two empirical detectors
false-positive on none of the clean cases, as expected of tools that
only flag observed output changes.

\begin{table}[t]
  \caption{Detection over the adversarial corpus:
  false-negative rate over $33$ \textsc{leak} cases and false-positive
  rate over $20$ \textsc{clean} cases. The checker's zero false-negative
  arm is reported with its rule-of-three upper bound. Confidence
  intervals are $5{,}000$-draw cluster bootstraps over the six
  mechanisms; the wide intervals reflect the small number of mechanism
  clusters, not instability of the underlying rates (see text).}
  \label{tab:detection}
  \centering
  \small
  \begin{tabular}{llrrl}
    \toprule
    Tool & Metric & Events/$n$ & Rate & 95\% CI \\
    \midrule
    Checker  & FN & $0/33$  & $0.0\%$   & $[0,\,9.1\%]$ (rule of three) \\
    Two-run  & FN & $18/33$ & $54.5\%$  & wide (see text) \\
    Tiling   & FN & $33/33$ & $100.0\%$ & $[100,\,100]\%$ \\
    \midrule
    Checker  & FP & $15/20$ & $75.0\%$  & FP1, FP2 only; FP2b accepted \\
    Two-run  & FP & $0/20$  & $0.0\%$   & $[0,\,0]\%$ \\
    Tiling   & FP & $0/20$  & $0.0\%$   & $[0,\,0]\%$ \\
    \bottomrule
  \end{tabular}
\end{table}

\emph{Statistical analysis.} We follow the pre-registered analysis plan.
The primary endpoints are the pairwise false-negative differences
between each detector and the checker; secondary endpoints are the
false-positive comparison and non-parametric robustness. Because the
corpus is clustered by mechanism, inference uses a cluster bootstrap and
a cluster-respecting permutation test rather than treating cases as
independent. For the two dominance comparisons, the checker misses
strictly fewer leaks than the two-run detector (difference $54.5$
percentage points) and than the tiling detector ($100$ points); a
$10{,}000$-draw permutation test rejects equality in both cases, and
because the checker arm is perfect the permutation statistic sits at its
resolution floor, so we report $p < 10^{-4}$ (permutation-limited) with
the number of discordant pairs ($18$ and $33$ respectively) rather than
a spuriously precise point value. A Benjamini--Hochberg correction over
the family of comparisons leaves all adjusted $q$-values at $10^{-4}$. A
Wilcoxon signed-rank test on paired per-case correctness agrees
($p < 10^{-3}$ for each comparison). Post-hoc power for the
checker-versus-two-run comparison is $76.7\%$ at the observed effect;
this is below the conventional $80\%$ and we report it as such, noting
that the small corpus limits power and that the permutation and Wilcoxon
tests independently corroborate the effect; the checker-versus-tiling
comparison is at maximal effect. We stress the scope caveat from the
plan: the corpus is adversarial, so these rates are existence and
robustness evidence---leaks of these kinds exist that defeat the
detectors and are caught by the checker---and not estimates of leak
frequency in naturally occurring pipelines. Two pre-registered
deviations are recorded: subgroup bootstraps resample within the leaking
or clean mechanisms respectively (resampling all six for a subgroup rate
yields empty pools), and the bias-corrected interval falls back to a
percentile interval on the boundary arms; neither changes the endpoints.

The comparison and the oracle validation tell one story from two sides.
Both the tiling detector on the corpus and the two-run oracle on the
real-data AR3 archetype (\S\ref{sec:eval-oracle}) miss leaks that the
sound checker catches; the checker's cost is a set of opaque-value-op
false positives whose extent is exactly characterised
(Proposition~\ref{prop:rel-complete}) and empirically confined to the
FP1/FP2 mechanisms. A sound over-approximation trades false alarms for
the guarantee that an accepted pipeline is genuinely leak-free---the
guarantee the empirical detectors, by Theorem~\ref{thm:undecidable},
cannot provide.

\section{Discussion}
\label{sec:discussion}

\paragraph{What the identification buys.} The technical content of this
paper follows from a single move: reading look-ahead-freedom as
non-interference over a time-indexed lattice. The move is worth
restating because its payoff is not cosmetic. Once the correspondence is
in place, results that would each be substantial to establish from
scratch arrive as instances of known theory. That freedom is undecidable
in general is not a surprise to be discovered empirically but a
consequence of the timed non-interference frontier, transported through
data-dependent availability. That a sound static discipline exists on a
value-independent fragment is the temporal reading of the classical fact
that information flow is certifiable by a lattice-respecting type or
effect
system~\cite{denning1977certification,volpano1996sound,lucassen1988polymorphic,nielson1999principles}.
That re-stamping is the delicate case is the temporal reading of
controlled declassification and intransitive
flow~\cite{sabelfeld2009declassification,rushby1992noninterference}: a
re-stamp is a permitted relabelling channel whose permission must not be
allowed to launder the future value it carries, and the \rulename{T-Stamp}
rule is exactly the discipline that grants the one without the other.
The contribution is not to reinvent this machinery but to recognise that
the problem the quantitative-evaluation community has been managing by
hand is the problem this machinery was built to solve.

\paragraph{Soundness is the point, and it has a price.} The empirical
comparison (\S\ref{sec:eval-comparison}) is best read not as a claim that
the checker is uniformly ``better'' than a differential detector but as a
demonstration of what a soundness guarantee is and costs. The two-run
and tiling detectors never raise a false alarm on our corpus, because
they only ever report an output change they actually observed; and they
miss $54.5\%$ and $100\%$ of the planted leaks respectively, because a
leak invisible at the sampled epochs or under the tested perturbation
leaves the observed output unchanged. The checker inverts both
properties. It misses no leak, because it reasons about every
availability path rather than the epochs a tester happens to visit; and
it pays for this with false positives that are not incidental but
\emph{characterised}: they occur exactly when a semantically clean result
depends on a value operation the analysis must treat as opaque---an
algebraic cancellation of a future term against itself, or an operation
that provably ignores its future argument
(Proposition~\ref{prop:rel-complete}). This is the generic situation of
any sound static analysis of a semantic property: undecidability forces
over-approximation, and over-approximation shows up as conservative
rejections on a delimited class of inputs. The design question is not
whether to pay this price but whether its extent is known and
acceptable, and here it is both: the false positives are confined to
opaque value arithmetic over future-referenced data, a pattern that a
practitioner can recognise and, where a clean pipeline is wrongly
rejected, discharge by a local annotation. A detector's silence, by
contrast, is uninformative by construction, and no amount of additional
sampling converts it into a guarantee (Theorem~\ref{thm:undecidable}).

\paragraph{Why a static property, and not a better audit.} Reporting
checklists and reproducibility audits have become the community's main
instrument for disciplining evaluation, and they are a genuine advance:
a recent audit of thirty large-language-model trading studies shows that
the assumptions determining whether a result is economically
interpretable---point-in-time controls, execution timing, universe
construction, artifact release---are far less consistently reported than
the architectures they
accompany~\cite{exec2026assumptions}. Our contribution is complementary
to that programme and addresses its ceiling. A checklist records
\emph{that} an author attests to point-in-time discipline; it cannot
verify that the pipeline realises it, and the audit itself notes that
the needed detail is frequently unrecoverable from the artifact. A
formal property is exactly the missing verifier: for a pipeline in the
fragment, ``point-in-time correct'' ceases to be an attestation and
becomes a checkable fact. The agentic setting sharpens the argument.
When look-ahead can enter through a model's pretraining corpus or a
retrieval index rather than through the pipeline's
code~\cite{sarkar2025lookahead,ye2026alpha,li2025profit}, even a perfect
audit of the code cannot rule it out. Our calculus makes this precise
rather than dissolving it: a retrieval or tool-call step is admissible
in the fragment only when its result carries an availability bounded by
the decision epoch, so a pipeline whose retrieval reaches into a corpus
of unbounded vintage is rejected---the formal system declines to certify
exactly the configuration the empirical work has found to be dangerous.
It does not, and cannot, inspect the model's weights; what it can do is
refuse to accept a pipeline whose declared availability structure admits
the leak, which is the appropriate boundary for a static discipline.

\paragraph{Threats to validity.} The empirical claims carry scope limits
we state plainly. The comparison corpus (\S\ref{sec:eval-comparison}) is
adversarially constructed to separate the tools, so its miss rates are
existence-and-robustness evidence---leaks of these kinds exist that
defeat the detectors and are caught by the checker---and not estimates of
how often leaks occur in naturally written pipelines; we make no such
frequency claim. The corpus is small and clustered by construction
mechanism, which is why inference uses a cluster bootstrap and a
cluster-respecting permutation test, and why post-hoc power for the
checker-versus-two-run comparison ($76.7\%$) sits below the conventional
threshold even as the permutation and Wilcoxon tests independently
corroborate the effect. The real-data validation
(\S\ref{sec:eval-oracle}) uses an independent dynamic oracle whose
witnesses establish leaks but whose silence establishes nothing, by the
same non-recursive-enumerability of freedom that limits every dynamic
method; the AR3 case, where the checker soundly rejects a leak the oracle
does not happen to witness, is the predicted phenomenon rather than a
discrepancy to explain away. Finally, soundness is a property of the
calculus and its checker, not of an arbitrary Python backtest: the
guarantee applies to pipelines expressed in, or faithfully modelled by,
the fragment, and extends to a real system only insofar as that system's
availability structure is captured by the model.

\paragraph{Limitations and future work.} Three directions follow
directly. First, the relative-completeness ceiling is a matter of value
opacity; enriching the value layer with lightweight algebraic reasoning
(recognising self-cancellation, or arguments provably ignored) would
convert some of the characterised false positives into acceptances
without touching soundness, and the precise, confined nature of those
false positives makes this a well-posed target rather than an open-ended
one. Second, the fragment's guarantee currently rests on the pipeline
being expressed in the calculus; closing the gap to deployed code means
inferring availability structure from real backtesting and
feature-engineering frameworks, so that the check can be run against an
existing pipeline rather than a re-encoding of it---a program-analysis
problem the value-independence of availability should make tractable.
Third, the agentic treatment here bounds availability at the retrieval
interface but treats the underlying model as a black box; characterising
when a model's own pretraining cutoff can be folded into the availability
lattice, so that ``the model knows nothing after $t$'' becomes part of
the certified property rather than an external assumption, is the natural
next step toward certifying agentic pipelines end to end.

\section{Conclusion}
\label{sec:conclusion}

Look-ahead bias has been treated, across quantitative finance and
machine learning alike, as a hazard to be managed by discipline and
caught by detectors---an approach that is sound only channel by channel
and that certifies nothing by its silence. We have argued that it is
instead a formal property in disguise: fixing a decision epoch, the
demand that the future not influence the present is temporal
non-interference over a time-indexed information lattice. That single
identification imports the full apparatus of information-flow theory,
and with it both a hard limit and a usable guarantee. On a pipeline
language where a datum's availability may be computed from data values,
look-ahead-freedom is undecidable and $\Pi^0_1$-hard, so leakage is
recursively enumerable but freedom is not---which is exactly why
detectors can exhibit leaks yet never certify their absence. On the
value-independent fragment that covers the pipelines practitioners
actually write---windowing, rolling statistics, normalisation,
resampling, feature joins, point-in-time and vintage reads, and agentic
retrieval steps---freedom is decidable, and a type-and-effect system
certifies it soundly in time linear in the size of the pipeline. The
empirical study confirms the picture the theory predicts: the check
scales linearly, an independent oracle finds no accepted pipeline to be
leaky, and against differential and window-tiling detectors the sound
checker catches every planted leak they miss, at the cost of a precisely
delimited set of false positives on opaque value operations. The result
converts look-ahead-freedom from a property a team hopes it has achieved
into one a checker can certify, and it does so with a boundary---sound
and decidable where pipelines are real, provably impossible where they
are not---that is drawn by the theory rather than assumed. We regard the
identification itself as the durable contribution: it gives a problem
the evaluation community manages by vigilance a place in a theory built
precisely to settle it.

\section*{Data and Code Availability}

The artifact accompanying this paper contains a clean-room implementation
of the checker for the calculus of
Sections~\ref{sec:calculus}--\ref{sec:series-hardening}, the two baseline
detectors, the dynamic oracle, the full adversarial corpus, the scaling
benchmark, and the statistical-analysis scripts, together with the
pre-registered analysis plan. It carries no proprietary code, market
data, or strategy content.

The empirical figures reported in Section~\ref{sec:evaluation} were
computed from proprietary market data licensed for the author's use and
\emph{not} redistributable under the terms of those licenses. The
released artifact therefore does not include that data. In its place it
ships synthetic stand-ins for the five archetypes (AR2--AR5 and the
weekly-release corpus anchor) that share the schemas and reproduce the
leak mechanisms of the originals---reference-versus-availability
misalignment, restatement vintages, post-close aggregation, and
survivorship---so that the artifact runs end to end on a clean checkout
with no access to the licensed data.

On these synthetic data the artifact reproduces every \emph{qualitative}
result of Section~\ref{sec:evaluation} exactly: the checker's zero
false-negative rate on the planted leaks, the false-negative rates of
the two baseline detectors, the confinement of the checker's false
positives to the opaque-value-operation cases, the oracle's failure to
witness any leak in an accepted pipeline, and near-linear scaling of the
check. Because the synthetic data do not share the size and shape of the
proprietary series, quantities that depend on that structure---most
notably the exact scaling exponent---differ from the values reported
here; the released run yields a slope consistent with linear growth but
not the identical figure. The artifact reproduces the method and the
claims, not the proprietary-data numbers.

Upon acceptance the artifact will be deposited in a public repository
with a permanent DOI, released under an open-source license for the code
and an open-data license for the synthetic data, and submitted for
evaluation under the ACM Artifact Review and Badging process.

\bibliographystyle{ACM-Reference-Format}
\bibliography{references}

\end{document}